\documentclass[11pt,english]{article}
\usepackage{pst-node}
\usepackage{pst-coil}
\usepackage[english]{babel}
\usepackage{multicol}
\usepackage[margin=1in]{geometry}
\usepackage{hyperref}
\usepackage[maxbibnames=99]{biblatex}
\usepackage{blindtext}
\hypersetup{
     colorlinks=true,
     linkcolor=blue,
     citecolor=blue,
 }
\usepackage{bm,xfrac,amsmath,amsthm,amssymb,soul, comment, xspace, mathtools}
\allowdisplaybreaks

\usepackage{cleveref}
\usepackage{tcolorbox}
\usepackage{enumitem}
\usepackage{csquotes}

\usepackage{booktabs}

\usepackage[T1]{fontenc} 
\usepackage{footnote,dsfont}
\usepackage{nicefrac,bm}
\usepackage[ruled,vlined,linesnumbered]{algorithm2e}
\usepackage{caption}
\usepackage{thmtools, float}
\usepackage{thm-restate}
\usepackage{graphicx}
\usepackage{pifont,multirow}
\usepackage{textcomp}
\usepackage{lmodern}
\usepackage{bold-extra}
\usepackage{multicol}

\usepackage{ifthen}

\bibliography{references}

\title{Covering a Few Submodular Constraints and Applications \thanks{Dept.\ of Computer Science, University of Illinois at Urbana-Champaign, Urbana, IL 61801, USA. Email: {\tt \{tbajpai, chekuri, poojark2\}@illinois.edu}. Supported in part by NSF grant CCF-2402667.}
}

\author{
Tanvi Bajpai
\and
Chandra Chekuri
\and
Pooja Kulkarni
}

\theoremstyle{plain}
\newtheorem{remark}{Remark}[section]

\newtheorem{lemma}{Lemma}[section]
\newtheorem*{lemma*}{Lemma}

\newtheorem{theorem}{Theorem}[section]

\newtheorem{claim}{Claim}[section]
\newtheorem*{claim*}{Claim}

\newtheorem{definition}{Definition}[section]

\newcommand{\A}{\mathcal{A}}
\newcommand{\F}{\mathcal{F}}

\newcommand{\Sets}{{\mathcal S}}
\newcommand{\Universe}{{\mathcal U}}

\newcommand{\I}{{\mathcal I}}
\newcommand{\J}{{\mathcal{J}}}

\newcommand{\lipc}{\ell}

\newcommand{\vecx}{\mathbf{x}}
\newcommand{\vecz}{\mathbf{z}}

\newcommand{\OPT}{{\sf OPT}}
\newcommand{\LP}{{\sf LP}}

\newcommand{\CCF}{{\sf CCF}}
\newcommand{\IP}{{\sf IP}}

\newcommand{\MSC}{{\sf MSC}}
\newcommand{\MBC}{{\sf MBC}}
\newcommand{\singleMSC}{\textsc{SingleRoundMSubCC}}

\newcommand{\Rset}{\mathbb{R}}
\newcommand{\Nset}{\mathbb{N}}
\newcommand{\scaleccf}{\left(\frac{e}{e-1}\right)\left( \frac{1}{1-\epsilon} \right)}

\newcommand{\E}{\mathbb{E}}

\newcommand{\FLMO}{{\sf FLMO}}
\newcommand{\FLRMO}{{\sf FLRMO}}

\newcommand{\keylemma}{\text{Rounding Lemma}} % macro for key lemma name; to be changed later
\newcommand{\h}{{high}}
\renewcommand{\l}{{low}}

\newcommand{\G}{\mathcal{G}}

\newcommand{\eps}{{\epsilon}}

\let\oldnl\nl% Store \nl in \oldnl
\newcommand{\nonl}{\renewcommand{\nl}{\let\nl\oldnl}}% Remove line number for one line

\newcommand{\note}[1]{{\color{red}\bf [#1]}}
\renewcommand{\nonl}{\renewcommand{\nl}{\let\nl\oldnl}}% Remove line number for one line
\long\def\symbolfootnote[#1]#2{\begingroup%
\def\thefootnote{\fnsymbol{footnote}}\footnote[#1]{#2}\endgroup}

\newcommand{\mypara}[1]{\medskip \noindent {\bf #1}}

\usepackage{titlesec}
\titlespacing{\section}{0pt}{*0.8}{*0.8}  % Adjusts spacing for \section
\titlespacing{\subsection}{0pt}{*0.7}{*0.7} % Adjusts spacing for \subsection
\titlespacing{\subsubsection}{0pt}{*0.6}{*0.4} % Adjusts spacing for \subsubsection

\begin{document}

\Crefname{algorithm}{Algorithm}{Algorithms}
\crefname{algorithm}{algorithm}{algorithms}

\date{}
\maketitle

\begin{abstract}
  We consider the problem of covering multiple submodular constraints.
  Given a finite ground set $N$, a cost function
  $c: N \rightarrow \mathbb{R}_+$, $r$ monotone submodular functions
  $f_1,f_2,\ldots,f_r$ over $N$ and requirements $b_1,b_2,\ldots,b_r$
  the goal is to find a minimum cost subset $S \subseteq N$ such
  that $f_i(S) \ge b_i$ for $1 \le i \le r$. When $r=1$ this is the
  well-known Submodular Set Cover problem. Previous work \cite{chekuri2022covering} considered
  the setting when $r$ is large and developed bi-criteria approximation algorithms, and approximation algorithms for the important special case when
  each $f_i$ is a weighted coverage function. These are fairly general models and
  capture several concrete and interesting problems as special cases. The approximation ratios
  for these problem are at least $\Omega(\log r)$ which is unavoidable when $r$ is
  part of the input.  In this paper, motivated
  by some recent applications, we consider the problem when $r$ is a
  \emph{fixed constant} and obtain two main results. 
  When the $f_i$ are weighted coverage functions from a deletion-closed set system we obtain a
  $(1+\eps)(\frac{e}{e-1})(1+\beta)$-approximation where $\beta$ is the approximation ratio for the underlying set cover instances via the natural LP. Second, for covering multiple submodular constraints we obtain a randomized bi-criteria approximation
  algorithm that for any given integer $\alpha \ge 1$ outputs a set $S$ such that $f_i(S) \ge (1-1/e^\alpha -\eps)b_i$ for each $i \in [r]$
  and $\E[c(S)] \le (1+\eps)\alpha \cdot \OPT$.
  These results show that one can obtain nearly as good an approximation for any fixed $r$ as what one would achieve for $r=1$.
  We also demonstrate applications of our results to implicit covering problems such as fair facility location.
\end{abstract}

\pagenumbering{arabic}

% \full{
% \input{ICALP/1_introduction}
% \input{ICALP/2_preliminaries}
% \input{ICALP/3_framework}
% \input{ICALP/4_MSC}
% \input{ICALP/5_CCF}
% %\input{ICALP/6_applications}
% }

% \conf{
\section{Introduction}
\label{sec:intro}
Covering problems are ubiquitous in algorithms and combinatorial optimization, forming the basis for a wide range of applications and connections.
Among the most well-known covering problems are {Vertex Cover} ({\sf VC}) and {Set Cover} ({\sf SC}). In {\sf SC} the input consists of a universe $\Universe$ of $n$ elements and a family $\Sets$ of subsets of $\Universe$. The objective is to find a minimum cardinality sub-collection $\Sets' \subset \Sets$ such that the union of the subsets in $\Sets'$ is (i.e. \emph{covers}) $\Universe$. {\sf VC} is a special case of {\sf SC} where $\Universe$ corresponds to the edges of a given graph $G=(V,E)$, and the sets correspond to vertices of $G$. In the cost versions of these problems, each set or vertex is assigned a non-negative cost, and the objective is to find a covering sub-collection with the minimum total cost. A significant generalization of {\sf SC} is the {Submodular Set Cover} ({SubmodSC}) problem where the input includes a normalized monotone submodular function $f:2^N \rightarrow \mathbb{Z}_+$, and the goal is to find a min-cost subset $S$ of $N$ such that $f(S) = f(N)$. Recall that a real-valued set function $f$ is submodular iff
$f(A) + f(B) \ge f(A \cup B) + f(A \cap B)$.

The aforementioned covering problems are known to be NP-Hard, but classical approximation algorithms offer strong guarantees. For {\sf SC}, the well-known Greedy algorithm proposed by Johnson~\cite{johnson1973approximation} achieves a $(1+\ln n)$-approximation and has also been shown to achieve a  $(1+ \ln (\max_i f(i))$-approximation for {\sf SubmodSC} by Wolsey \cite{wolsey1982analysis}. These results are essentially tight, assuming $P\neq NP$ \cite{feige1998threshold}. Nevertheless, {\sf VC} admits a $2$-approximation, and many other special cases of {\sf SC} and {\sf SubmodSC} also admit constant-factor approximations. Moreover, settling for \emph{bi-criteria} approximations allows for additional flexibility and improvements. For {\sf SubmodSC}, a slight modification to the standard Greedy algorithm yields the following tradeoff: for any integer $\alpha \ge 1$, we may obtain a set $S \subseteq N$ such that $f(S) \ge (1-\frac{1}{e^\alpha})f(N)$ and $c(S) \le \alpha \OPT$. % Notably, the modified algorithm retains a polynomial runtime that is independent of $\alpha$ \note{added; check/cite}.

Another related covering problem of importance is the {Partial Set Cover} ({\sf PartialSC}) problem, in which the input includes an additional integer $b$, and the goal is to find a minimum-cost collection of sets that cover \emph{at least} $b$ elements. A similar formulation extends to the submodular setting where the goal is to find a set $S$ such that $f(S) \ge b$. Partial covering problems are particularly useful as they naturally model outliers for various settings, and approximation algorithms for these problems have been well studied. It is straightforward to show that {\sf PartialSC} is equivalent to {\sf SC} in terms of approximation. In the submodular case, this equivalence is even clearer: we can consider the truncated submodular function $f_b$ where $f_b(S) = \min\{f(S), b\}$. In contrast, understanding the approximability of {Partial Vertex Cover} ({\sf PartialVC}) is more nuanced. While a $2$-approximation for {\sf PartialVC} is known \cite{bshouty1998massaging}, it is not straight forward to see this.

%\mypara{Recent applications in colorful and fair covering:} 
In recent years, emerging applications and connections, particularly to fairness, have sparked interest in covering and clustering problems involving multiple groups or colors for elements. The goal in these problems is to ensure that some specified number of elements from each group or color are covered. The partial covering problems discussed above are a special case where the number of groups is one. %\note{Keep one of the following examples. Chandra - skip the search example, it reads as a proforma to me.} 
In the {Colorful Vertex Cover} problem ({\sf ColorVC}), multiple partial covering constraints are imposed. The input consists of a graph $G=(V,E)$ and $r$ subsets of edges $E_1,E_2,\ldots,E_r$ where $E_i \subseteq E$. Each subset represents the edges of color $i$, and an edge can have multiple colors. Each color $i$ has requirement $b_i \le \vert E_i \vert$ and the objective is to find a minimum-cost set of vertices that covers at least $b_i$ edges from each $E_i$. 
%This problem can be viewed as having multiple partial vertex cover constraints.  For instance, consider the problem of diversity in search. Suppose we search for ``computer scientists'' in a set of images. We want the output results to show computer scientists with diversity across race, religion, geography, gender and other demographic factors. Typically these groups are much smaller in comparison to the total number of images in the database and can be considered a constant. This problem can be generalized as a partial covering problem with constraints across multiple groups. 
Bera et al.~\cite{bera2014approximation} were the first to consider (a generalization of) {\sf ColorVC}, obtaining an  $O(\log r)$-approximation. This result is essentially tight when $r$ is large and part of the input, and one can show that {\sf SC} reduces to {\sf ColorVC} (see \cite{bera2014approximation}). More recently,
Bandyapadhyay et al.~\cite{bandyapadhyay2021fair} considered {\sf ColorVC} when $r$ is a \emph{fixed constant} and achieved a $(2+\eps)$-approximation for any fixed $\eps > 0$. 

Chekuri et al.~\cite{chekuri2022covering} explored a broader class of problems involving the covering of multiple (submodular) functions and various special cases. Their general framework not only recovers the $O(\log r)$-approximation from  \cite{bera2014approximation} but also provides results for a variety of geometric set cover problems, and combinatorial optimization problems such as facility location and clustering. However, while powerful and general, their framework is designed for cases where $r$ (i.e., the number of constraints) is part of the input. As a result, it is inherently limited to $\Omega(\log r)$-approximation due to above mentioned reduction from {\sf SC} to {\sf ColorVC}. In a recent work,  Suriyanarayana et al.~\cite{suriyanarayana2024improved} considered the Joint Replenishment Problem ({\sf JRP}), a fundamental problem in supply chain optimization. {\sf JRP} and its variants can be viewed as a covering problems where the objective is to satisfy demands over a time horizon for multiple items using orders that incur joint costs; readers can refer to \cite{suriyanarayana2024improved} for the formal details of the {\sf JRP} problem. While constant-factor approximations for {\sf JRP} are well-established \cite{bienkowski2015approximation}, Suriyanarayana et al.~\cite{suriyanarayana2024improved} introduced a colorful version of {\sf JRP} ({\sf CJRP}), similar in spirit to {\sf ColorVC}. In {\sf CJRP}, demands are grouped by color, and the goal is to satisfy at least $b_i$ demands from each color class. Chekuri et al.~\cite{chekuri2022covering}'s framework can yield an $O(\log r)$-approximation for {\sf CJRP}; however, \cite{suriyanarayana2024improved}
focused on the setting where $r$ is a fixed constant. They developed an intricate algorithm which achieves a $(2.86+\eps)$-approximation for any fixed $\eps > 0$. Note that the approximation ratio does not depend on $r$ but the running time is exponential in $r$. There are several settings in approximation for covering and packing problems where one has a fixed number of constraints or objectives and one obtains similar trade offs. 

%\paragraph{Covering a constant number of submodular constraints:} 
These recent developments inspire us to explore whether the general framework in \cite{chekuri2022covering}, which applies to a wide variety of covering problems, can be adapted to the setting of a fixed number of submodular covering constraints. The goal is to removed the dependence of the approximation ratio on $r$. Specifically, we consider two such settings: (1) The Multiple Submodular Covering ($\MSC$) problem and (2) Colorful Set Cover and a generalization.

\paragraph{Multiple Submodular Covering ($\MSC$).} The input to the $\MSC$ problem consists of $r$ monotone submodular functions $f_i:2^N \rightarrow \mathbb{Z}_+$ over a finite ground set $N$ (provided as value oracles), a non-negative cost function $c: N \rightarrow \mathbb{R}_+$, and non-negative integer requirements $b_1,b_2,\ldots,b_r$. The objective is to find a minimum-cost set $S \subseteq N$ such that $f_i(S) \ge b_i$ for each $i \in [r]$. Note that when $r=1$, this corresponds to well-known Submodular Covering problem (SubmodCover) \cite{wolsey1982analysis}. The Greedy algorithm (modified slightly) yields the following trade off for SubmodCover: for any integer $\alpha \ge 1$ the algorithm outputs a set $S$ such that $f(S) \ge (1-1/e^\alpha)b$ and $c(S) \le \alpha \OPT$. We note that this trade off is essentially tight for any fixed $\alpha$ via the hardness result for Set Cover \cite{feige1998threshold}.
It is easy to reduce multiple submodular covering constraints to a single submodular covering constraint \cite{Har-PeledJ23,chekuri2022covering}, but in doing so one is limited to a logarithmic approximation for the cost to cover all constraints --- it is not feasible to distinguish the individual constraints via this approach.
In contrast,
\textcite{chekuri2022covering} obtained the following bi-criteria approximation result for MSC: there is an efficient randomized algorithm that outputs a set $S \subseteq N$ such that $f(S) \ge (1-1/e-\eps)b_i$ for each $i \in [r]$ and $E[c(S)] \le O(\frac{1}{\eps} \log r)$. In this paper we ask whether it is possible to achieve bi-criteria bounds for MSC instances with a fixed number of constraints that match those obtainable for single submodular constraint. We prove the following theorem and its corollary which together affirmatively answers this (modulo a small dependence on $\eps$).

\begin{restatable}{theorem}{mainresultmsc}
%\begin{theorem}
\label{thm:main-result-msc}
There is a randomized polynomial-time algorithm that given an instance of the Multiple Submodular Covering Constraints problem with fixed $r$ and $\eps > 0$ outputs a set $S \subseteq N$ such that (i) $f(S) \ge (1-1/e-\eps)b_i$ for all $i \in [r]$ and (ii) $E[c(S)] \le (1+\eps)\OPT$ (where $\OPT$ is the cost of the optimal solution to the instance). 

%For an instance of the Multiple Submodular Covering Constraints problem with a constant number of submodular constraints (i.e. $r \in O(1)$) and for $\eps > 0$ there is a randomized polynomial-time algorithm that outputs a set 

%\end{theorem}
\end{restatable}

\begin{restatable}{corollary}{mainresultmscgen}
%\begin{corollary}
\label{coll:main-result-msc}
    For any fixed integer $\alpha > 1$ and $\eps$, Algorithm \ref{alg:msc-single} (from \Cref{thm:main-result-msc}) can be used to construct a set $S \subseteq N$ such that $f(S) \ge (1-1/e^\alpha-\eps)b_i$ for all $i \in [r]$ and $E[c(S)] \le \alpha(1+\eps)\OPT$. 
\end{restatable}

\paragraph{Colorful Set Cover and generalizations.} MSC provides a general framework for covering, however the limitation of the generality is that we are limited to a bi-criteria trade off even for $r=1$. On the other hand a special cases such as ColorVC with fixed $r$ admits a constant factor approximation. Can we identify when this is possible? Indeed, prior work of Inamdar and Varadarajan \cite{InamdarV18} showed a general setting when positive results for SC can be extended to PartialSC ($r=1$).
We describe the setting first.   
A class of SC instances (i.e. set systems), $\mathcal{F}$, is called \emph{deletion-closed} if for any given instance in the class, removing a set or a point 
results in an instance that also belongs to $\mathcal{F}$. Many special cases of SC are naturally deletion-closed. For example, VC instances are deletion-closed as are many geometric SC instance (say covering points by disks in the plane).  Suppose one can obtain a $\beta$-approximation to SC on instances from $\mathcal{F}$
via the natural LP relaxation. For example, VC has $\beta = 2$ and there are many geometric settings where $\beta = O(1)$  or $o(\log n)$ \cite{chan2012weighted,InamdarV18}.
\textcite{InamdarV18} formalized the concept of deletion-closed set systems and obtained a $2(\beta+1)$-approximation for PartialSC on a set system from deletion-closed $\mathcal{F}$. Their framework gave a straightforward way to approximate geometric partial covering problems by building on existing results for SC. \textcite{chekuri2022covering} extended this framework to address deletion-closed instances of ColorSC problem (which is analogous to the previously discussed ColorVC). Formally, the input for ColorSC consists of a set system from $\mathcal{F}$, costs on the sets, and $r$ sets $\Universe_1,\ldots,\Universe_r$ where each $\mathcal{U}_i$ is a subset of the universe $\Universe$, and non-negative integer requirements $b_1,\ldots,b_r$. The objective is to find a minimum-cost collection of sets such that for each $i \in [r]$, at least $b_i$ elements from $\Universe_i$ are covered. In this setting, \textcite{chekuri2022covering} achieved two key results. First, when $r=1$, they improved the approximation bound from \cite{InamdarV18} to $\frac{e}{e-1}(\beta+1)$. 
Second, for general $r$ they obtained an $O(\beta + \log r)$-approximation which generalized \cite{bera2014approximation} for ColorVC.  Here we consider this same 
setting but when $r$ is a fixed constant and obtain the following result.

\begin{theorem}\label{thm:intro-ccf}
For any instance of ColorSC from a $\beta$-approximable deletion-closed set system and any fixed $\eps > 0$
there is a randomized polynomial-time algorithm that yields a $\left(\frac{e}{e-1}\right)(1 + \beta)(1 + \epsilon)$-approximate feasible solution.
\end{theorem}

\textcite{chekuri2022covering} prove a more general result that applies to covering integer programs induced by a deletion-closed set systems (this is similar to the context considered by \textcite{bera2014approximation} for ColorVC). 
This generality is significant in applications such as facility location and clustering. We obtain the preceding result
also in this more general setting; we defer the formal description of the problem setting to Section~\ref{sec:prelim} and discuss the theorem in Sec~\ref{sec:ccf}.

\paragraph{Implications and Applications.}

At a high level, our two main theorems show that approximation results for covering problems with a single submodular constraint, i.e. where $r = 1$, can be extended 
to the case with a fixed number of constraints.  We believe that these results, in addition to being independently interesting, are broadly applicable and valuable due to the modeling power of submodularity and set covering. However, one should expect additional technical details in applying these results to certain applications such as facility location, clustering, {\sf JRP}, and others. In this paper, we prove our main technical results and lay down the general framework. We also demonstrate how the framework can be applied and adapted. We discuss these applications next.
%In particular, consider first the following immediate application.

\noindent
{\em Colorful covering problems:}
Recall that {\sf ColorVC} admits a $(2+\eps)$-approximation for fixed $r$ \cite{bandyapadhyay2021fair}. Now consider the problem of covering points in plane by disks with costs.
This well-studied geometric set cover problem admits an $O(1)$ approximation via the natural LP \cite{chan2012weighted,varadarajan2010weighted}. 
\cite{InamdarV18,chekuri2022covering} obtained a $\left(\frac{e}{e-1}\right)(1 + \beta)$ approximation for {\sf PartialSC} version of this problem ($r=1$)
and 
\cite{chekuri2022covering} obtained an $O(\beta + \log r)$-approximation for the colorful version of this problem (arbitrary $r$). In this paper we obtain a $\left(\frac{e}{e-1}\right)(1 + \beta)(1 + \epsilon)$ approximation for any fixed $r$. In effect all the {\sf PartialSC} problems considered in \cite{InamdarV18,chekuri2022covering} for various geometric settings admit essentially the same approximation for the colorful versions with fixed $r$. The colorful problems capture outliers and fairness constraints, and thus, as long as the number of such constraints is small, the approximation ratio is independent of $r$. Our bi-criteria approximation for MSC, in a similar vein, allows tight trade offs in cost vs coverage for any fixed number of constraints. See \cite{Har-PeledJ23,chekuri2022covering} for an application to splitting point sets by lines where a bi-criteria approximation is adequate.

\noindent
{\em Facility location with multiple outliers:} Facility location and clustering problems have long been central in both theory and practice, with a wide variety of objective functions studied in the literature. At a high level, the goal in facility-location problems is to select a subset of facilities and assign clients to them in a way that balances assignment cost with facility opening cost.
Several objectives and variants are studied in the literature.
Motivated by robustness considerations there is also extensive work on these problems when there are outliers. The goal in such problems is to solve the underlying clustering or facility location problem where certain number of points/clients can be omitted; equivalently the goal is to connect/cluster at least a given number of points.
These problems do not fall neatly in the framework of (partial) covering problems. Nevertheless it is possible to view them as \emph{implicit} covering problems; in some settings one needs to consider an exponential sized set system. Thus, although the problems cannot be directly cast in as $\MSC$ or $\CCF$, some of the ideas can be transferred with additional technical work. 
\cite{chekuri2022covering} used this perspective and additional ideas to obtain $O(\log r)$-approximation for two facility location problems with multiple outlier classes. In this paper, we
consider these two problems when $r$ is a fixed constant and obtain constant factor approximations. The first one involves facility location to minimize the objective of sum of radii; it is easier to adapt the ideas for $\CCF$ to this problem and we discuss this in \Cref{sec:radii}.

Our main new technical contribution is the second problem related to the outlier version of the well-known uncapacitated facility location problem ({\sf UCFL}). In {\sf UCFL} the input consists of a facility set $F$ and a client set $C$ that are in a metric space $(F \cup C,d)$. Each facility $i \in F$ has an opening cost $f_i \ge 0$ and connecting a client $j$ to a facility $i$ costs $d(i,j)$. The goal is to open a subset of the facilities $S \subseteq F$ to minimize $\sum_{i \in S}f_i + \sum_{j \in C} d(j,S)$ where $d(j,S)$ is the distance of $j$ to the closest open facility. This classical problem has seen extensive work in approximation algorithms; the current best approximation ratio is $1.488$ \cite{li20131}. Charikar et al. \cite{charikar2001algorithms} studied the outlier version of this problem (under the name robust facility location)
and obtained a $3$-approximation; in this version the goal is to connect at least some specified number $b \le |C|$ of clients, and this corresponds to having a single color class to cover. In \cite{chekuri2022covering}, the generalization of this to $r$ outlier classes was studied under the name Facility Location with Multiple Outliers ($\FLMO)$. In this paper we prove the following result.
\begin{restatable}{theorem}{mainresultflmo}
%\begin{theorem}

Let $\beta$ be the approximation ratio  for {\sf UCFL} via the natural LP relaxation. Then there is a randomized polynomial-time algorithm that given an instance of $\FLMO$ with fixed $r$ and $\eps > 0$ yields a $\frac{e}{e-1}(\beta + 1 + \eps)$-approximate solution.
\end{restatable}

We remark that we use LP-based algorithms for {\sf UCFL} as a black-box and have not attempted to optimize the precise ratio. Using $\beta=1.488$ from \cite{li20131} we get an approximation factor of $3.936(1+\epsilon)$. We also note that the case of $r$ being fixed requires new technical ideas from those in \cite{chekuri2022covering}. We recall the non-trivial work of \cite{suriyanarayana2024improved} who consider {\sf CJRP} for fixed $r$ by adapting the ideas from \cite{chekuri2022covering} for an implicit problem. One of our motivations for this paper is to obtain a general framework for $\CCF$ for fixed $r$, and to showcase its applicability to concrete problems. 
%Overall, we believe that the framework we present here is flexible and powerful and can be used in a variety of settings with some problem-specific adaptations and optimizations.

\paragraph{Summary of Technical Contributions.} Our algorithms follow a streamlined framework (outlined in \Cref{sec:rf-framework}) that builds upon the ideas in \textcite{chekuri2022covering}, while introducing additional new ingredients to obtain improved approximations when $r$ is fixed. The framework consists of four stages; at a high level, these stages can are: guess high-cost elements, solve a continuous relaxation, randomly round the fractional relaxation, and greedy fix in a post-processing step to ensure that constraints are appropriately satisfied.  The main new ingredient is the rounding procedure employed in the third stage, which is centered around a  rounding lemma (\Cref{lem:additive-bound}). The lemma is based on the concentration bound for independent rounding  of a fractional solution to the multilinear relaxation of a submodular function. The improved approximation for fixed $r$ is based on obtaining a strong concentration bound. In order to obtain this we need to ensure that the underlying submodular function has a Lipschitz property. Our new insight is that one can use a greedy procedure to select a small number of elements to ensure the desired Lipschitz property. However this means that the algorithms will incur an additive approximation for selecting a constant number (that depends only on $r$ and $1/\eps$) of elements. We are able to convert this additive approximation into a relative approximation by guessing a sufficiently large number of elements in the first stage. 
This is a high-level overview and there are several non-trivial technical details in the algorithms, some of which are inherited from \cite{chekuri2022covering}.
We also note that although both algorithms adhere to the high-level four-stage framework, there are several important differences in the details for the two problems which are explained in the respective sections.

\paragraph{Other Related Work.}
This paper is inspired by recent work
on covering with a fixed number of constraints \cite{bandyapadhyay2021fair,suriyanarayana2024improved}
and previous work on covering submodular constraints for large $r$ \cite{bera2014approximation,chekuri2022covering}. We refer to
those papers for extensive discussion of motivating applications, related problems, and technical challenges.
In the context of Set Cover and {SubmodSC}, an important problem is Covering Integer Program (CIPs) that lies in between.
CIPs have been studied extensively for several years starting with \cite{Dobson1982}. The introduction of KC inequalities \cite{CFLP} led
to the first $O(\log m)$-approximation by Kolliopoulos and Young
\cite{KY}. Recent work \cite{CHS,CQ19} has obtained sharp bounds.
We observe that CIPs with fixed number of constraints admit an easy PTAS due to the underlying LP where
one can use properties of a basic feasible solution. In some sense that is also the way \cite{bandyapadhyay2021fair}
obtain their algorithm for {ColorVC} with fixed $r$. However, this approach does not generalize
to our setting since our submodular constraints are more complex (even for say covering points by disks).
Multiple outliers and related covering constraints have been recently also been studied in
clustering, partly motivated by fairness, and there are several different models. We only mention one problem, namely Colorful $k$-Center \cite{bandyapadhyay2019constant,jia2022fair,anegg2022technique,ceccarello2024fast}. 
This problem can be also be viewed as a covering problem with multiple constraints.
The difference is that in $k$-Center the cost is not allowed to be violated (can use only $k$ centers) while
the sets corresponding to the balls of optimum radius which are allowed to expand since the objective is to 
minimize the radius. Our problems
allow cost to be approximated while the sets are not fungible and the constraints are much more complex. 
Optimization problems with submodular \emph{objectives} have been been an active
area of research both in maximization and minimization setting,
and there is extensive literature. Some recent successes, especially in the approximation setting,
have relied on sophisticated mathematical programming techniques and we too rely
on these techniques even though our problems involve submodular \emph{constraints}. 

\mypara{Organization:} \Cref{sec:prelim} provides the formal definitions for $\MSC$ and $\CCF$, as well as some relevant background on submodularity. In \Cref{sec:rf} we describe the $\keylemma$ (\Cref{lem:additive-bound}), which is our key ingredient. Then we provide an overview of our algorithmic framework. \Cref{sec:msc} describes our algorithm for \Cref{thm:main-result-msc}. \Cref{sec:ccf} describes our algorithm for $\CCF$ which generalizes the result given in \Cref{thm:intro-ccf}. In \Cref{sec:ccf-apps}, we outline two new applications of the $\CCF$ framework to facility location problems. The appendix contains the proof of \Cref{coll:main-result-msc}.
\section{Notation and Preliminaries}\label{sec:prelim}
\subsection{Problem Definitions}

Below we provide the formal problem statements for Multiple Submodular Covering Constraints ($\MSC$) and Covering Coverage Functions ($\CCF$) problems \cite{chekuri2022covering}.

\begin{definition}[$\MSC$]\label{def:msc}
    Let $N$ be a finite ground set, and we are given $r$ monotone submodular functions $f_i: 2^N \to \mathbb{R}_+$ for $i \in [r]$ each with corresponding demands $b_i \in \mathbb{R}_+$. Additionally, we are given a non-negative cost function $c : N \to \mathbb{R}_+$. The goal is to find a subset $S \subseteq N$ that solves the following optimization problem:
    \begin{align*}
        \min_{S \subseteq N} ~~& c(S) \\
        \text{s.t.}~~ & f_i(S) \geq b_i \quad \forall i \in [r].
    \end{align*}
\end{definition}

%\paragraph{Multiple Submodular Constrained Covering ($\MSC$):} Our most general problem, which we call $\MSC$ consists of (1) A ground set, $N$ with a cost function $c : N \rightarrow \mathbb{R}^+$, (2) $r$ polymatroids $f_1, \ldots, f_r$ over $N$ and $r$ integers $b_1, \ldots, b_r$. We want to select a minimum cost set of elements, $S \subseteq N$ so that we satisfy all $f_i(S) \geq b_i$ for all $i \in [r]$. Throughout, we assume $r$ is a fixed constant. % \in O(1)$.

\begin{definition}[$\CCF$]\label{def:ccf}
    The $\CCF$ problem consists of a set system $(\Universe, \Sets)$ where $\Universe$ is the universe of $n$ elements, $\Sets = \{S_1, \ldots, S_m\}$ is a collection of $m$ subsets of $\Universe$. Each set $S \in \Sets$ has a cost associated with it and we are given a set of inequalities $A\vecz \geq \bf{b}$ where $A \in \mathbb{R}_{\geq 0}^{r \times n}$. The goal is to optimize the integer program given in \ref{ip:fair-covering}.
    %The $\CCF$ problem is a special case of the $\MSC$ problem where each element of the ground set is a subset of elements of some universe $\Universe$. We use $\Sets$ to denote this set system and assume $|\Sets| = m$, $|\Universe| = n$. The collection of sets $\Sets$ has a cost function $c: \Sets \rightarrow \mathbb{R}^+$ associated with it. The input also consists of $r$ weighted coverage functions $f_1, \ldots, f_r$ and $r$ integers $b_1, \ldots, b_r$. Our problem is to pick a minimum cost collection of sets $\mathcal{C} \subseteq \Sets$ such that $f_i(\mathcal{C}) \geq b_i$ for all $i \in [r]$. Again, here $r$ is a constant. These constraints can be expressed using an $r \times n$ matrix $A$.
\end{definition} 
 \begin{minipage}[h!]{0.40\textwidth}
 \begin{equation}\label{ip:fair-covering}\tag{$\IP$-$\CCF$}
 \large \begin{array}{cccc}
     & \min \sum_{i \in [m]} c_ix_i \\ 
 \text{s.t.}     & \sum_{i : j \in S_i} x_i \geq z_j \\ 
   &  A \vecz  \geq   b & \\ 
   & z_j \leq 1 & \text{for all } j \in [n] \\ 
    & x_i \in \{0, 1\}  & \text{for all } i \in [m]
 \end{array}
 \end{equation}
 \end{minipage}%
 \hfill
 \begin{minipage}[h!]{0.40\textwidth}

 \begin{equation}\label{lp:fair-covering}\tag{$\LP$-$\CCF$}
 \large \begin{array}{cccc}
     & \min \sum_{i \in [m]} c_ix_i \\ 
 \text{s.t.}     & \sum_{i : j \in S_i} x_i \geq z_j \\ 
   &  A \vecz  \geq   b & \\
   & z_j \leq 1 & \text{for all } j \in [n] \\ 
    & x_i \in [0,1]  & \text{for all } i \in [m]
 \end{array}
 \end{equation}
 \end{minipage}
\vspace{.1cm}

In this program, the variables $x_i$ for $i \in [m]$ represent the indicator variable for whether the set $i$ is selected. The variables $z_j$ for $j \in [n]$ are indicators for whether the element $j$ is covered. If entries of $A$ are in $\{0,1\}$ then we obtain ColorSC as a special case.
% We can reformulate the $\CCF$ problem as an $\MSC$ problem in the following way. 

\paragraph{CCF as a special case of MSC.} Recall that the constraints for this problem are represented by a matrix $A \in \mathbb{R}_{\geq 0}^{r \times n}$. We view these as submodular functions as follows: for the $i^{th}$ constraint, given by the $i^{th}$ row of matrix, we define a submodular function $f_i(\cdot) : 2^{\Sets} \rightarrow \mathbb{R}_{\geq 0}^+$ as follows: $f_i(\Sets) = \sum_{j \in [n]}A_{i,j}z_j$ where $z_j$ is a variable indicating whether $j \in \cup_{S \in \Sets} S$ i.e., whether $j$ is covered by a set in the collection $\Sets$ of sets, and $A_{i,j}$ is the entry in $i^{th}$ row and $j^{th}$ column of matrix $A$. This is a submodular function since it is a weighted coverage function with weights coming from the matrix $A$.
In \Cref{sec:ccf}, we will use either the matrix form or the submodular function form as convenient.

In this paper we assume that input instances of each problem will have a fixed number of constraints, i.e., $r$ is taken to be some positive fixed constant. 
%\paragraph{Covering Coverage Functions ($\CCF$):} The $\CCF$ problem is a special case of the $\MSC$ problem where each element of the ground set is a subset of elements of some universe $\Universe$. We use $\Sets$ to denote this set system and assume $|\Sets| = m$, $|\Universe| = n$. The collection of sets $\Sets$ has a cost function $c: \Sets \rightarrow \mathbb{R}^+$ associated with it. The input also consists of $r$ weighted coverage functions $f_1, \ldots, f_r$ and $r$ integers $b_1, \ldots, b_r$. Our problem is to pick a minimum cost collection of sets $\mathcal{C} \subseteq \Sets$ such that $f_i(\mathcal{C}) \geq b_i$ for all $i \in [r]$. Again, here $r$ is a constant. These constraints can be expressed using an $r \times n$ matrix $A$.

\subsection{Submodularity} \label{sec:prelim-submod}
For a submodular function, $f(\cdot)$ defined on ground set $N$, we will use $f_{\mid A}(\cdot)$ to denote the submodular function that gives marginals of $f$ on set $A \subseteq N$ i.e., $f_{\mid A}(S) = f(S \cup A) - f(A)$. Further, for any $A \subseteq N$ and $e \in N$, we will use $f(e \mid \A)$ to denote the marginal value of $e$ when added to $A$.

\paragraph{Multilinear Extensions of Set Functions.} \label{sec:prelim-mle}

Our main rounding algorithm (discussed in \Cref{sec:rf}) makes use of the multilinear extension of set functions. Here are some relevant preliminaries.

\begin{definition}[multilinear extension] \label{def:mle}
    The multilinear extension of a real-valued set function $f: 2^N \to R$, denoted by $F$, is defined as follows: For $x \in [0,1]^N$ \[ F(x) = \sum_{S\subseteq N} f(S) \prod_{i \in S} x_i \prod_{j \in S} (1-x_j). \]

    Equivalently, $F(x) = \E[f(R)]$ where $R$ is a random set obtained by picking each $i \in N$ independently with probability $x_i$. 
\end{definition} % \label{def:mle}

\paragraph{Lipschitzness and Concentration Bounds.}\label{sec:prelim-lips}

One benefit of utilizing the multilinear extension that is relevant to our algorithm is a concentration bound. Before we state that bound, we define \emph{Lipschitzness} of a submodular function.

\begin{definition}[$\ell$-Lipschitz]\label{def:lips}
    A submodular function $f: 2^N \rightarrow \Rset_+$ is $\ell$-Lipschitz if $\forall i \in N$ and $A \subset N$: \[ |f(A \cup \{i\}) - f(A) | \leq \ell. \] When $f$ is monotone, this amounts to $f(\{i\}) \leq \ell$. 
\end{definition}
\begin{lemma}[\textcite{vondrak2010note}]\label{thm:vondrak}
    Let $f: 2^N \to \Rset_+$ be an $\ell$-Lipschitz monotone submodular function. For $\vecx \in [0,1]^N$, let $R$ be a random set where each element $i$ is drawn independently using the marginals induced by $\vecx$. Then for $\delta \geq 0$,
    \[ Pr[f(R) \leq (1-\delta) F(\vecx) ] \leq e^{-\frac{\delta^2 F(\vecx)} {2\ell}}. \] 
\end{lemma}
%In our algorithm and analysis, we work with constraints of the form, $f(S) \geq b$ where $f$ is a submodular functions and $b$ is some demand. In this context, it is additionally beneficial to determine whether a function is $\ell$-Lipschitz continuous with respect to some subset of the ground set $N$. That is, if the given $f$ is not $\ell$-Lipschitz continuous, it may be the case that $f$ restricted to some domain may be. We provide a formal definition of this property below.  

%It is unlikely that these functions have the Lipschitz property as defined in \Cref{def:lops}, since in the context of our optimization problems, it may not be natural to consider $f$ that are $\ell$-Lipschitz for some meaningful choice of $\ell$. Instead, we may In this context, $f$ may not satisfy $\ell$-Lipschitzness.  these constraints might not inherently satisfy $\ell$-Lipschitzness as is. However, a function can be Lipschitz when restricted to a certain domain. %We therefore define next the notion of \emph{Constraint $f$ being Lipschitz with respect to domain $S$} as follows

%\begin{definition}[$\ell$-Lipschitz constraint with respect to domain $A$]\label{def:lips-cons}
%    Given a submodular constraint: ``Find $S \subseteq N$ such that $f(S) \geq b$'', we say that $f$ is $\ell$-Lipschitz with respect to domain $A \subseteq N$ if for all $e \in A$, $f(e) \leq \ell \cdot b$.
%\end{definition}

\paragraph{Discretizing costs and guessing:} Our problems involve minimizing the cost of objects. Via standard guessing and scaling tricks, if we are willing to lose a $(1+o(1))$ factor in the approximation, we can assume that all costs are integers and are polynomially bounded in $n$. In particular this implies that there are only $\text{poly}(n)$ choices for the optimum value. We will thus assume, in some algorithmic steps, knowledge of the optimum value with the implicit assumption that we are guessing it from the range of values. This will incur a polynomial-overhead in the run-time, the details of which we ignore in this version.

\iffalse
\begin{remark}\label{rem:opt}
    Our algorithms and analysis for results assume that we know what the cost of the optimal solutions to instances of $\MSC$ or $\CCF$ are. Note that we can determine this value by assuming that the cost is some integer between $1$ and $n^3$ (where $n$ is the size of the ground set being considered) and using binary search to determine the optimal cost. For instances where this assumption does not hold, one may facilitate some clever scaling and enumeration to ensure that this is the case. \note{Chandra - you had said you'd add to this} %   \note{Do we need the last sentence?}
\end{remark}
\fi
\section{High-Level Algorithmic Framework and \keylemma}\label{sec:rf} % todo: change name

%\subsection{Chandra}

\begin{comment}

\begin{lemma}\label{lem:additive-bound-single-constraint}
 Let $N$ be a finite set and let $c: N \rightarrow \mathbb{R}_+$ be non-negative cost function on $N$ and let $c_{\max} = \max_j c_j$.
 Let $f:2^N \rightarrow \mathbb{R}_+$ be a monotone submodular function. Suppose we have a fractional solution $\vecx \in [0,1]^N$
 such that $F(\vecx) \ge b$ and $\sum_j c_j x_j \le B$. Then there is a set $S$ such that $c(S) \le B + c_{\max}$ and $f(S) \ge b$.
\end{lemma}
\begin{proof}
We sketch a proof. Do pipage rounding. Fix two fractional variables $x_1$ and $x_2$ without loss of generality. If we change $x_1$ to $x_1+\eps_1$ and $x_2$ to $x_2 - \eps_2$ then $F(x)$ changes to $\eps_1 \partial F/\partial x_i - \eps_2 \partial F/\partial x_2 + \delta$ where $\delta \ge 0$. Based on $c_1$ and $c_2$ and $\eps_1 \partial F/\partial x_i$ and $\partial F/\partial x_2$ we should be able to move $x$ such that cost does not increase and $F$ does not decrease. Thus, we can keep doing this until we only have one fractional variable. At this point we simply make that $1$. Computational aspects are a bit messy. There may be another proof saying that Greedy itself will give what we want from the statement. 
\end{proof}
\end{comment}

Our algorithms for approximating $\MSC$ and $\CCF$ share a unified framework. The framework relies on a \keylemma~(\Cref{lem:additive-bound}), which provides a method to convert a fractional solution to an integral one in a way that preserves key properties of the original fractional solution while incurring only a  minimal additive cost. In this section, we begin by presenting and proving our \keylemma~(\Cref{sec:rf-rounding}). We then provide a high-level overview of the algorithmic framework (\Cref{sec:rf-framework}) built upon this \keylemma.%, which forms the foundation for our results on $\MSC$ and $\CCF$, derived in the subsequent sections.

% All our algorithms follow a common framework, which we explain in this section. A key part of this framework is rounding a fractional covering with only a small extra cost. In Subsection \ref{sec:rounding}, we introduce the main lemma for rounding, and in Subsection \ref{sec:framework}, we describe the overall framework. Readers interested in the broad framework can use Lemma \ref{lem:additive-bound} as a black box and skip ahead to Section \ref{sec:framework}. \tanvi{todo: make better preamble.}

\subsection{\keylemma}\label{sec:rf-rounding}

A key step in our framework is the ability to round fractional solutions to integral ones while maintaining important guarantees. Our \keylemma~(stated below) provides a formal method for achieving this. It shows that for $r$ submodular functions ($r \in O(1)$), a fractional solution satisfying certain constraints with respect to their multilinear extensions can be (randomly) rounded in polynomial time to an integral solution. This rounding incurs only a small (constant sized) additive increase in cost and ensures that the values of the submodular functions are sufficiently preserved with high probability.

\begin{lemma}[\keylemma] \label{lem:additive-bound} % todo: change name
     Let $N$ be a finite set, $c: N \rightarrow \mathbb{R}_+$ be a non-negative cost function and let $c_{\max} = \max_{j\in N} c_j$. Let $f_1, \ldots, f_r$ be monotone submodular functions over $N$. For each $f_i$, let $F_i$ be its corresponding multilinear extension. Suppose we are given a fractional point $\vecx \in [0,1]^N$ such that $\sum_j c_j x_j \le B$ and $F_i(\vecx) \ge b_i$ for each $i \in [r]$. Then, for any fixed $\eps > 0$ there is a randomized algorithm that outputs a set $S \subseteq N$ such that 
     \begin{enumerate}[label=(\roman*)]
         \item $\E[c(S)] \le B + r\lceil \frac{2 \ln (\sfrac{r}{\epsilon}) \ln (\sfrac{1}{\epsilon})}{\epsilon^2}\rceil c_{\max}$, and 
         \item $\forall i \in [r]$ $f_i(S) \ge (1-\epsilon)b_i$ with probability at least $1 - \sfrac{\epsilon}{r}$.
     \end{enumerate}
    % (i) $\E[c(S)] \le B + r\lceil \frac{2 \ln (\sfrac{r}{\epsilon}) \ln (\sfrac{1}{\epsilon})}{\epsilon^2}\rceil c_{\max}$, and    (ii) $\forall i \in [r]$ $f_i(S) \ge (1-\epsilon)b_i$ with probability at least $1 - \sfrac{\epsilon}{r}$.
\end{lemma}

To prove this lemma, we first establish the following helpful claim, which ensures that after selecting a subset of elements, a monotone function becomes Lipschitz continuous (\Cref{def:lips}), with a Lipschitz constant $\lipc$ that depends on the number of elements chosen.

\begin{claim}\label{clm:lipschitz}
Let $N$ be a finite set and $f$ be a non-negative monotonic\footnote{Note that we do not need submodularity in this claim.} function over $N$. For any $\epsilon > 0$, $b > 0$ and $\lipc > 0$, there exists a polynomial time algorithm that returns a set $S \subseteq N$ of cardinality at most  $\lceil \frac{1}{\lipc}{\ln(\frac{1}{\epsilon})}\rceil$ such that one of the following conditions hold: 
    \begin{enumerate}[label=(\roman*)]
         \item $f(S) \geq (1-\epsilon)b$,
         \item $\forall e \in N \setminus S$, $f(e \mid S) < \ell \cdot (b - f(S))$.
     \end{enumerate}
     
  % \[\textnormal{(i) } f(S) \geq (1-\epsilon)b\textnormal{ or (ii) }\forall e \in N \setminus S\textnormal{, }f(e \mid S) < \ell \cdot (b - f(S))\]
\end{claim}
\begin{proof}
   Consider the following greedy procedure to construct $S$:  Initialize \( S \) as the empty set. While there exists an element \( e \in N \setminus S \) satisfying 
\( f(e \mid S) \geq \ell \cdot (b - f(S)) \) and \( f(S) < (1 - \epsilon)b \), add \( e \) to \( S \). 
Finally, return \( S \).

We prove that in at most $\lceil \frac{1}{\lipc}\ln \frac{1}{ \epsilon} \rceil$ iterations of the while loop, the set $S$ must satisfy $(i)$ or $(ii)$. Note that trivially, the returned set $S$ satisfies one of these conditions. In the rest of the proof we will bound the cardinality of $S$.% when this is (first) satisfied. 

Suppose that we never satisfy $(ii)$ in $\lceil \frac{1}{\lipc}\ln \frac{1}{\epsilon} \rceil $ iterations. We will show using induction that we must then satisfy $(i)$. In particular, let $S^{(t)}$ denote the set $S$ after $t$ iterations of the while-loop in the given procedure (where $S^{(0)} = \emptyset$), and let $s_{t+1}$ be the element that is to be added to $S$ in the $(t+1)$-th iteration. Now, using induction on $t$, we will show that for all $t \geq 0$ the following inequality will hold 
\[ f(S^{(t)}) \geq (1 - (1-\ell)^t)b. \]

% \{s_k \mid k \in [t+1]\} <= commenting here inc ase we wanna change it back, idk
For $t = 0$, $S$ is the empty set, hence $f(S^{(0)}) = 0$, and $(1 - (1-\ell)^0)b = 0$. Assume $f(S^{(t)}) \geq (1 - (1-\ell)^t)b$ for some fixed $t$. We will show that the inequality must hold for $t + 1$. To see this, observe that if $(ii)$ does not hold, then there exists an element, $s_{t+1}$ such that $f(s_{t+1} \mid S^{(t)}) \geq \lipc (b - f(S^{(t)}))$. Therefore we will have,
    \begin{align*}
        f(S^{(t+1)}) &= f(s_{t+1} \mid S^{(t)}) + f(S^{(t)}) \\
        &\geq \ell(b - f(S^{(t)}) + f(S^{(t)}) \\
        &= \ell b + (1-\ell)f(S^{(t)}) \\
        &\geq \ell b + (1-\ell)(1-(1-\ell)^t)b \\ 
        &= (1-(1-\ell)^{t+1})b
    \end{align*}
    Therefore, for $t = \lceil \frac{\ln \eps}{\ln(1-\ell)} \rceil$, $(1-(1-\ell)^{t})b \geq (1-\eps)b$ holds. Thus, we are guaranteed that after selecting $\frac{\ln \eps}{\ln(1-\lipc)} \leq \lceil \frac{1}{\lipc}\ln \frac{1}{\epsilon}\rceil$ elements, if $(ii)$ is not met $(i)$ holds. %This completes the proof.
\end{proof}

Using \Cref{clm:lipschitz}, we prove the $\keylemma$.%  below. 

\begin{proof}[Proof of \keylemma~(Lemma \ref{lem:additive-bound})]
    Consider the following algorithm. For a value of $\ell$ that we will determine later, for each $i \in [r]$, create a set $S_i$ according to the greedy algorithm in claim \ref{clm:lipschitz}. We are therefore guaranteed that for each constraint $i \in [r]$ $f_i(S_i) \geq (1 - \epsilon)b_i$ or $f_i(e \mid S_i) \leq \ell(b_i - f(S_i))$ for all $e \in N \setminus S_i$. %is $\ell$-Lipschitz with respect to $N \setminus S_i$ (Definition \ref{def:lips-cons}). 
    Note that $S_i$'s may not be disjoint. Next, perform independent randomized rounding with marginal probabilities given by $\vecx$ to find a set $S'$. Let $S = \cup_{i \in [r]} S_i \cup S'$. 

     We first prove that each constraint is satisfied to a $(1-\epsilon)$ approximation with probability $1 - \sfrac{\epsilon}{r}$. Fix any constraint $i \in [r]$. Suppose the set $S_i$ returned from Claim \ref{clm:lipschitz}'s algorithm satisfies $f_i(S_i) \geq (1-\epsilon) b_i$, then we are done. Otherwise, we know that for all $e \in N \setminus S_i$, $f(e \mid S_i) < \ell (b-f_i(S_i))$. Denote by $f_{i\mid S_i}$ the submodular function that is defined by $f_{i\mid S_i}(X) \coloneqq f_i(X \mid S_i)$. Let $b_i' = b_i - f_i(S_i)$. We have that for all $e \in N \setminus S_i$, $f_i(e \mid S_i) \leq \ell b_i'$ and therefore, $f_{i\mid S_i}(e) \leq \ell b'_i$ for all $e \in N \setminus S_i$. Then, when we (independently) randomly round using marginal probabilities given by $\vecx$, we get using \Cref{thm:vondrak} for the function $f_{i \mid S_i}$\footnote{Note that the Lemma \ref{thm:vondrak} applies when we round elements in $N \setminus S_i$. Here we round all elements in $N$ and, since $f_i$ is monotone, this is guaranteed to have a higher value than rounding some of them, therefore we can safely apply the concentration bound.},
    
    \begin{align}\label{eqn:conc}
        \Pr[f_{i \mid S_i}(S') \leq (1-\epsilon)F_{i \mid S_i}(x)] &\leq e^{-\frac{\epsilon^2 F_{i\mid S_i}(x)}{2 \ell b'_i}}
        \end{align}
        Now, $b_i' = b_i - f_i(S_i) \leq F_i(\vecx) - f_i({S_i}) = F_i(\vecx) - F_i(S_i)$. Now, 
        \begin{align*}
            F_i(\vecx) - f_i(S_i) &= \sum_{S \sim \vecx} f_i(S)Pr[S] - f_i(S_i) &&\text{(By Definition \ref{def:mle})}\\
            &=\sum_{S \sim \vecx} f_i(S) Pr[S] - \sum_{S \sim \vecx} f_i(S_i)Pr[S]  &&\text{(Since the probabilities sum to one)}\\
            &= \sum_{S \sim \vecx} (f_i(S) - f_i(S_i)) Pr[S] \\
            &\leq \sum_{S \sim \vecx} (f_i(S \cup S_i) - f_i(S_i)) Pr[S] \\
            &= \sum_{S \sim \vecx} f_{i \mid S_i}(S) Pr[S] &&\text{(Using monotonicity of $f_i$)}\\
            &= F_{i \mid S_i}(\vecx) &&\text{(By Definition \ref{def:mle})}
        \end{align*}
        \begin{comment}
        \begin{align*}
        F_{i \mid S_i}(\vecx) &= \sum_{S \sim \vecx} f_{i \mid S_i}(S) Pr[S] \\
        &\geq \sum_{S \sim \vecx} (f_i(S) - f_i(S_i)) Pr[S] \\
        & = \sum_{S \sim \vecx} f_i(S) Pr[S] - \sum_{S \sim \vecx} f_i(S_i)Pr[S] \\
        &= \sum_{S \sim \vecx} f_i(S)Pr[S] - f_i(S_i) \\
        &= F_i(\vecx) - f_i(S_i). 
        \end{align*}
        \end{comment}
        Therefore,
        $b_i' \leq F_{i \mid S_i}(\vecx)$. Substituting this in Equation \ref{eqn:conc},
        \begin{align*}
         \Pr[f_{i \mid S_i}(S') \leq (1-\epsilon)F_{i \mid S_i}(x)] \leq e^{-\frac{\epsilon^2b_i'}{2 \ell b'_i}}
        \end{align*}
        Now, choosing $\ell = \frac{\epsilon^2}{2 \ln (\sfrac{r}{\epsilon})}$
        \begin{align*}
        \Pr[f_{i \mid S_i}(S') \leq (1-\epsilon)F_{i \mid S_i}(x)] &\leq e^{-\frac{\epsilon^2 }{2 \epsilon^2} \cdot 2\ln (\sfrac{r}{\epsilon})} = \frac{\epsilon}{r}.
    \end{align*}    
    Therefore, with probability at least $1 - \sfrac{\epsilon}{r}$, 
    \begin{align*}
        f_i(S) &= f_i(S_i) + f_i(S' \mid S_i)\\
        &\geq f_i(S_i) + (1-\epsilon)F_{i \mid S_i}(x)\\
        &\geq (1-\epsilon)F_i(x)\\
        &= (1-\epsilon)b_i.
    \end{align*}
     Now, we will look at the cost of the sets chosen. To get the required probability bound, we need to choose $\ell = \frac{\epsilon^2}{2\ln(\sfrac{r}{\epsilon})}$. Therefore, following claim \ref{clm:lipschitz}, the size of set $S_i$ for each $i \in [r]$ is at most $\lceil \frac{1}{\lipc} \ln \frac{1}{\epsilon}\rceil = \lceil \frac{2 \ln (\sfrac{r}{\epsilon}) \ln (\sfrac{1}{\epsilon})}{\epsilon^2}\rceil$. Therefore, together the total cardinality of $\cup_{i \in [r]}S_i$ is bounded by $ r\lceil \frac{2 \ln (\sfrac{r}{\epsilon}) \ln (\sfrac{1}{\epsilon})}{\epsilon^2}\rceil$. Finally, the expected cost of $S'$ is $B$. Together, we get $\E[c(S)] \leq B + r\lceil \frac{2 \ln (\sfrac{r}{\epsilon}) \ln (\sfrac{1}{\epsilon})}{\epsilon^2}\rceil c_{\max}$.

  It is easy to see that the algorithm runs in polynomial time.
\end{proof}

\subsection{Overview of Algorithmic Framework}\label{sec:rf-framework}

%\note{check this; otherwise can revert back to abridged version from conference drafts} 
At a high level, $\MSC$ and $\CCF$ are optimization problems that involve selecting a low-cost subset of \emph{objects} that is able to sufficiently satisfy a set of covering constraints. A natural approach to tackling these problems is to consider their relaxations and allow for objects to be chosen fractionally. However, these relaxations are known to have an unbounded integrality gap \cite{chekuri2022covering}. Our algorithmic framework, which consists of four main stages, describes how to get around this integrality gap when there are a fixed number of covering constraints. 

% At a high level, $\MSC$ and $\CCF$ are problems that involve selecting a low-cost set of \emph{objects} to satisfy a set of submodular constraints. For both problems, we can relax the problem in natural way to select the elements fractionally. However, these natural relaxations have unbounded integrality gap. Our algorithmic framework shows how to get around this. It consists of four main stages, which we outline below:\note{needs to be written for elements (not sets); TODO post final intro}

\begin{itemize}
    \item \textbf{Stage 1: Guessing Highest-Cost Objects from the Optimal Solution.} Let $N$ denote the initial set of objects that we are to select from. We begin by pre-selecting the $L$ highest-cost objects from the optimal solution. Here, $L$ is chosen to be some sufficiently large constant number that depends on $r$, i.e., the number of submodular constraints, which is assumed to be fixed. To construct this subset, we ``guess,'' i.e. enumerate all possible $L$-sized subsets of objects from $N$ that could potentially make up this desired set. After doing so, we use this subset of objects to construct a \emph{residual instance} of the input instance in which all objects with costs higher than those guessed are removed and where the objects in the subset are considered ``chosen'' as part of the outputted solution. % $N'$ denote the set of objects that remain. We note that the precise details of how $L$ is set and how the residual instance is defined will vary from problem to problem.

 % This pruning step is safe because sets with costs larger than the highest-cost sets in the optimal solution cannot be part of the optimal solution.
    \item \textbf{Stage 2: Constructing a Fractional Solution.} Let $N'$ denote the set of objects that remain after Stage 1. Next, we construct a fractional solution \( \vecx \in [0,1]^{N'} \) for the residual instance so that it satisfies certain properties outlined in the statement of our $\keylemma$ (\Cref{lem:additive-bound}). % This step involves solving a LP that ensures the fractional solution respects the given covering constraints.
    \item \textbf{Stage 3: Rounding the Fractional Solution $\vecx$ to Construct $S$.} In this stage, the fractional solution $\vecx$ is rounded into an integral solution $ S \subseteq N'$ using our $\keylemma$ (\Cref{lem:additive-bound}). This rounding ensures that the integral solution sufficiently satisfies the covering constraints with high probability while incurring a small additive cost.
    \item \textbf{Stage 4: Greedy Fixing of $S$.} Finally, since the solution $S$ obtained in Stage 3 only satisfies the covering constraints with high probability, additional adjustments are needed to ensure all constraints are satisfied. Our framework achieves this by greedily selecting additional objects from those remaining until all constraints are sufficiently satisfied. %The choice of greedy algorithm will vary based on the original problem formulation. %Different greedy algorithms may be needed for different problems. 
\end{itemize}

%These stages together form the core of our framework, allowing us to effectively handle covering problems like $\MSC$ and $\CCF$ subject to a fixed number of covering constraints. 
While both our algorithms for $\MSC$ (\Cref{sec:msc}) and $\CCF$ (\Cref{sec:ccf}) follow the high-level framework described above, the execution of certain stages differs. The most notable difference occurs in Stage 2, where the construction of the fractional solution to be used in the $\keylemma$ is tailored for the specific problem and target result. Details on the execution of each of these stages can be found in their respective sections.

\section{Multiple Submodular Covering Constraints} \label{sec:msc}

In this section we consider $\MSC$ (\Cref{def:msc}) and prove \Cref{thm:main-result-msc} (restated below). We begin by discussing some technical components that are needed for our algorithm in \Cref{sec:msc-pre}. In \Cref{sec:msc-alg} we describe an algorithm based on the framework outlined in \Cref{sec:rf-framework}. Then, in \Cref{sec:msc-analysis}, we provide the analysis of our algorithm and the proof of \Cref{thm:main-result-msc}. In later sections, we will show how this algorithm can be extended to derive bi-criteria approximation results with improved coverage (\Cref{coll:main-result-msc}).

Recall that an instance $\J$ of $\MSC$ consists of a finite ground set $N$, a non-negative cost function $c : N \to \mathbb{R}_+$, and $r$ monotone submodular functions $f_1, f_2, \ldots, f_r$, each with an associated requirement $b_1, b_2, \ldots, b_r$, where $r$ is some fixed constant. The objective of $\MSC$ is to find a subset $S \subseteq N$ of minimum cost such that $f_i(S) \geq b_i$ for all $i \in [r]$. We will assume (without loss of generality) that $f_i(N) = b_i$ for all $i \in [r]$ by truncating $f_i$ to $b_i$.% as larger values of $b_i$ can be reduced by truncating the function $f_i$ to $\min\{b_i, f_i(\cdot)\}$.

\mainresultmsc*

\subsection{Algorithm Preliminaries}\label{sec:msc-pre}

Before proceeding with the overview of our algorithm, we present some definitions and results that are used as subroutines in our algorithm. 

\subsubsection{Constructing residual $\MSC$ instances}

Recall that in Stage 1 of our algorithmic framework (\Cref{sec:rf-framework}) constructs a residual instance after guessing a subset of elements from the optimal solution. Below is a definition which describes how to construct a residual and cost-truncated residual $\MSC$ instances with respect to subset of elements $A$ from the initial ground set. 

\begin{definition}[Residual and Cost-Truncated Residual $\MSC$ Instances]\label{def:residual-msc}
    Given an $\MSC$ instance $\J = (N, \{f_i, b_i\}_{i \in [r]}, c)$ and a set $A \subseteq N$, a \textbf{residual instance of $\J$ with respect to $A$} is $\J' = (N', \{f'_i, b'_i\}_{i \in [r]}, c)$ where $N' := N \setminus A$, for all $i \in [r]$, $f'_i(\cdot) := f_{i \mid A}(\cdot)$ and $b'_i := b_i - f_i(A)$. We do not change $c$ since the cost values remain unchanged, but assume its domain is restricted to $N'$. A \textbf{cost-truncated residual instance} of $\J$ with respect to $A$ is a residual instance $\J' = (N', \{f'_i, b'_i\}_{i \in [r]}, c)$ in which elements in $N'$ that have a higher cost than every element in $A$ are also removed. More precisely, $N' = \{e \mid e \in N \setminus A \textnormal{ and s.t. }\forall e' \in A~ c(e) \leq c(e') \}$. 
    % \begin{itemize}
    %     \item If $A \subseteq S^*$ and contains the highest cost elements in $S^*$ (i.e. $\forall e_1 \in A$ and $\forall e_2 \in S^* \setminus A$ $c(e_1) \geq c(e_2)$), then additionally update $N'$ by removing all elements that have costs higher than those in $A$. More precisely, let $N' = \{e \mid e \in N \setminus A \textnormal{ and s.t. }\forall e' \in A~ c(e) \leq c(e') \}$.\footnote{Note that we may not do this for residual instances with respect to arbitrary $A \subseteq S^*$, since removing higher-cost elements in those cases may result in the removal of high-cost elements from $S^*$, which would result in a residual instance which would have an optimal solution of higher cost than the original instance.} \note{Chandra - this is part of def since its easier to explain here than mess with or overload notation later}
    % \end{itemize}
\end{definition}

\subsubsection{Continuous relaxation of $\MSC$ and approximation results}\label{sec:msc-pre-cr}

In Stage 2 to construct the desired fractional $\vecx$, we use the continuous relaxation of $\MSC$, stated below as a feasibility program in \ref{lp:mp-submod-relax}. For an instance $\J = (N, \{f_i, b_i\}_{i \in [r]}, c)$, the objective of \ref{lp:mp-submod-relax} is to find a fractional point $\vecx \in [0,1]^N$ whose cost is at most $C$ and where the value of the multilinear extension of each constraint $f_i$, denoted by $F_i$, at $\vecx$ satisfies the demand $b_i$. 

Note that for $C \geq \OPT$, where $\OPT$ is the cost of the optimal solution to the $\MSC$ instance, \ref{lp:mp-submod-relax} is guaranteed to be feasible. %that does not exceed $C$.  the following continuous optimization problem will have a feasible solution of cost $C$:

\begin{equation}
        \begin{array}{r@{}c@{}l}\label{lp:mp-submod-relax}\tag{\MSC-Relax}
            ~~~& \sum_{j}c_j x_j \leq C\\
            \text{s.t.} &~~  F_i(\vecx) \geq b_i & ~~\forall i \in [r] \\
            & ~~~~~~~\vecx \geq 0 
        \end{array}
    \end{equation}

Finding a feasible solution to the above continuous optimization problem is NP-Hard. The following result can be used to get an approximate solution. % \tanvi{check if wording of thm makes sense}

\begin{theorem}[\textcite{chekuri2010randomized,chekuri2015multiplicative}]\label{thm:mle}
   For any fixed $\eps > 0$, there is a randomized polynomial-time algorithm that, given an instance of \ref{lp:mp-submod-relax} and a value oracle access to the submodular functions $f_1,..., f_r$, with high probability, either correctly outputs that the instance is not feasible or outputs an $\vecx$ such that \[ (i)~\sum_j c_jx_j \leq C \textnormal{ and } (ii)~F_i(\vecx) \geq (1-1/e-\eps)b_i ~\forall i \in [r].\]
\end{theorem}

\subsubsection{Greedy algorithm to fix $\MSC$ constraints}\label{sec:msc-pre-gr}
In Stage 4 of our $\MSC$ algorithm, we need to fix unsatisfied constraints, and for this we can focus on a single submodular covering constraint.
The relevant optimization problem is the following.
\[ \min c(T) \textnormal{ s.t. }f(T) \geq b\]
Suppose the optimum cost for this is $C^*$. Our goal is to show that we can find a set $S$ such that
$f(S) \ge (1-1/e)b$ and $c(S) \le C^*$. For this purpose we guess $C^*$ and recast the problem 
as a submodular function \emph{maximization} problem subject to a knapsack constraint with budget $C^*$.
\[ \max f(T)\textnormal{ s.t. } c(T) \leq C^* \] 
Note that there is a feasible solution to this maximization problem of value at least $b$.
We can now use the following result to accomplish our goal.

\begin{lemma}[\textcite{sviridenko2004note}]\label{lem:greedy}
    There is a $(1-1/e)$-approximation for monotone submodular function maximization subject to a knapsack constraint.
\end{lemma}

\iffalse
\begin{lemma}[\textcite{sviridenko2004note}]\label{lem:greedy}
    Consider an instance of monotone submodular function maximization subject to a knapsack constraint. Let $Z$ be the optimum value for the given knapsack budget $C$. There is a polynomial time algorithm that outputs a set $T$, such that $c(T) \leq C$ and $f(T) \geq (1-1/e)Z$ where $Z$ is the optimum value.
\end{lemma}
\fi

\subsection{Algorithm Overview}\label{sec:msc-alg}

We provide pseudo-code for our algorithm in Algorithm \ref{alg:msc-single} and demarcate the various steps involved in each of the framework stages outlined in \Cref{sec:rf-framework}. We additionally provide more detailed descriptions of how each of the four stages in the context of our $\MSC$ approximation. % to help highlight the intuition and inherent simplicity behind our approach. 

\begin{algorithm}[h!]
\caption{Pseudocode for Bi-criteria Approximation for $\MSC$ ($\singleMSC$)}
\label{alg:msc-single}
\DontPrintSemicolon
  \SetKwFunction{Define}{Define}
  \SetKwInOut{Input}{Input}\SetKwInOut{Output}{Output}
  \Input{An instance of $\MSC$ problem denoted by $\J = (N, \{ f_i, b_i \}_{i=1}^r, c)$ where $f_i$'s are normalized, parameter $\eps$}
  \Output{Solution $S \subseteq N$ that satisfies \Cref{thm:main-result-msc}}
  \BlankLine 
  $S_{pre} \gets$ guess the $L$ highest cost elements from an optimal solution to $\J$\; \label{step:msc-guess-sets}
  Eliminate all elements with higher costs than those of $S_{pre}$, and create a cost-truncated residual instance of $\J$ with respect to $S_{pre}$ called $\J':= (N', \{ f'_i, b'_i \}_{i=1}^r, c)$.\tcp*[r]{Stage 1} \label{step:msc-res}
  Using the algorithm from \Cref{thm:mle}, obtain vector $\vecx \in \{0,1\}^{N'}$ as an approximate solution of \ref{lp:mp-submod-relax} for instance $J'$. \tcp*[r]{Stage 2}
  Using \Cref{lem:additive-bound}, round $\vecx$ to obtain a set $R \subseteq N'$. \tcp*[r]{Stage 3}
   \For{$i \in [r]$}{ 
   $T_i \gets \emptyset$ \tcp*[r]{Stage 4}
   \If{$f'_i(R) < (1-1/e - 2\epsilon)b'_i$}{
    $T_i \gets$ Elements using Greedy Algorithm by \cite{sviridenko2004note} s.t. $f'_i(T_i) \geq (1-\sfrac{1}{e})b'_i$. \label{alg:msc-single:fix} 
    }
   }
   Let $T := \bigcup_{i \in [r]} T_i$\;
   \Return{$S_{pre} \cup R \cup T$} \label{alg:msc-single:return}
\end{algorithm}

%\subsubsection

\paragraph{Stage 1: Choosing $L$ and constructing $\J'$.} % \label{sec:msc-alg-stage1}

 The first stage of our algorithm involves guessing sufficiently many (i.e. $L$) high-cost elements from the optimal solution, denoted by $S^*$, that we use to restrict our input instance. To properly set $L$, we must first define a few additional variables, whose role will become clearer in later stages and in our analysis. First, we define $\eps'$ and $\eps''$ such that $0 < \eps',\eps'' \leq \eps$ and $(1-\eps'')(1-\sfrac{1}{e}-\eps') \leq (1 - \sfrac{1}{e} - \eps)$. Next, we define $\ell := \frac{ (\eps'')^2}{2\ln \sfrac{r}{\eps''}}$. With this, we finally set $L := r \cdot \lceil \frac{1}{\eps''}(\frac{1}{\lipc}\ln \frac{1}{\eps''})\rceil$. We additionally assume that $|S^*| > L$, otherwise we can simply return $S_{pre}$ as our final solution.

Now, upon guessing (i.e. enumerating) $S_{pre}$, we construct a cost-truncated residual instance of $\J$ with respect to $S_{pre}$ called $\J'$ per \Cref{def:residual-msc}. Since $\J'$ is cost-truncated, we know that $N'$ will not contain any element whose cost is higher than elements in $S_{pre}$.  % Note that since, by definition, $S_{pre} \subseteq S^*$ (the optimal solution to $\J$), the residual ground set will not contain any element with cost higher than elements selected for $S_{pre}$. 

%\subsubsection
\paragraph{Stage 2: Constructing a fractional solution $\vecx$ using the continuous relaxation of $\J'$.} 
%\label{sec:msc-alg-stage2}

To construct our fractional solution $\vecx$, we use the continuous relaxation of $\J'$, as defined in \Cref{sec:msc-pre-cr}. Let $\OPT$ and $\OPT'$ denote the cost of an optimal solution to $\J$ and $\J'$, respectively. For now, suppose that we know beforehand what $\OPT$ is (see \Cref{sec:prelim-submod} for more details). Since we also assume $S_{pre} \subseteq S^*$, it should be clear that $\OPT' = \OPT - c(S_{pre})$. Using \Cref{thm:mle} we obtain a vector $\vecx$ s.t. $\sum_j c_jx_j \leq \OPT'$ and $\forall i \in [r]$ $F'_i(\vecx) \geq (1 - 1/e - \eps')b'_i$. Recall that $\eps'$ was defined in Stage 1. % (recall that $\eps'$ was defined as part of Stage 1 of this algorithm). 

%\subsubsection
\paragraph{Stage 3: Rounding $\vecx$ to construct set $R$.}
%\label{sec:msc-alg-stage3}

Now that we have obtained a fractional solution $\vecx$ that has certain desirable properties, we use our Rounding Lemma (\Cref{lem:additive-bound}) to round $\vecx$ to an integral solution $R$. The precise usage of the Rounding Lemma for this algorithm is as follows: Given $N'$, $c$, $\{f'_i\}_{i=1}^{r}$ and their corresponding multi-linear extensions $\{F'_i\}_{i=1}^r$, our fractional point $\vecx \in [0,1]^{N'}$ satisfies $\sum_j c_jx_j \leq \OPT'$ and for all $i \in [r]$, $F'_i(\vecx) \geq b''_i$ (where $b''_i := (1-\sfrac{1}{e}-\eps')b'_i$). The algorithm invoked in the Rounding Lemma can thus be used to output a set $R \subseteq N'$ s.t. (i) 
$\E[c(R)] \leq \OPT' + r\lceil \frac{2 \ln (\sfrac{r}{\epsilon''}) \ln (\sfrac{1}{\epsilon''})}{(\epsilon'')^2}\rceil c_{\max}$ (where $c_{\max}$ is the cost of the most expensive element in $N'$) and (ii) $\forall i \in [r]~f'_i(R) \geq (1-\eps'')b''_i$ with probability at least $1 - \sfrac{\eps''}{r}$. Recall that $\eps''$ was defined in Stage 1.  

%\subsubsection
\paragraph{Stage 4: Greedily Fixing Unsatisfied Constraints.} %\label{sec:msc-alg-stage4}
For each constraint $f'_i$ that has yet to be sufficiently satisfied by $R$, i.e., $f'_i(R) < (1-1/e-2\eps)b'_i$, use the procedure outlined in \Cref{sec:msc-pre-gr} to reformulate the problem of finding a corresponding $T_i$ to ``fix'' (i.e. satisfy) it as a submodular maximization problem subject to a single knapsack constraint. The budget for the knapsack constraint will be $\OPT'$. Now, using the greedy procedure from \Cref{lem:greedy}, we obtain, for each unsatisfied $f'_i$, a set $T_i$ s.t. $c(T_i) \leq \OPT'$ and $f'_i(T_i) \geq (1-\sfrac{1}{e})b'_i$. 

Letting $T := \bigcup_{i \in [r]} T_i$, our algorithm returns the elements in $S_{pre} \cup R \cup T$. In the following subsection, we show how this output will satisfy the coverage and cost bounds in \Cref{thm:main-result-msc}

% \begin{remark}\label{rem:opt-gr}
%     As was the case in \Cref{rem:opt-cr}, to determine the best value to set the knapsack budget for each constraint (i.e. $\OPT'_i$), we can simply use binary search and \Cref{lem:greedy}. \tanvi{how much detail}.
% \end{remark}
% Now, let $\OPT'_i$ denote the cost of the optimal solution to the above (that is, $c(T^*_i)$ for the optimal solution $T^*_i$). With this, we may reformulate the above as the following submodular maximization problem subject to a single knapsack constraint, allowing for the use of the greedy procedure from \Cref{lem:greedy} to construct each $T_i$. \[ \max f'_i(T_i)\textnormal{ s.t. } c(T_i) \leq \OPT'_i \] 

%In the proof of \Cref{lem:tjs} we will see that, after using binary search to determine how to fix $\OPT'_i$, we can apply \Cref{lem:greedy} to show that the greedy procedure employed by \textcite{sviridenko2004note} will allow us to construct a set $T_i$ that will provide sufficient coverage for $f'_i$'s and ultimately $f_i$'s.

% will describe the precise usage of the above result in our analysis in the following subsection. For now, we note that it will be used for any $f'_i$ s.t. $f'_i(R) < (1-1/e-2\eps)b'_i$

\subsection{Analysis}\label{sec:msc-analysis}

We now analyze Algorithm \ref{alg:msc-single} and prove \Cref{thm:main-result-msc}. We begin by proving the following simple claim that allows us to determine $\OPT'$ using $\OPT$. Recall that $\OPT$ and $\OPT'$ denote the cost of the optimal solutions to $\J$ and $\J'$ respectively. %the cost of the optimal solution to $\J'$ using that of $\J$. Recall that $\J'$ is the cost-truncated residual instance of $\J$ with respect to $S_{pre}$, i.e. the set of elements guessed in Stage 1. %, where $S_{pre}$ is assumed to be made up of the $L$ highest cost elements optimal solution to $\J$.    

% Recall that $\OPT$ and $\OPT'$ denote the cost of the optimal solutions to $\J$ and $\J'$< respectively, and $S_{pre}$ is the subset of elements guessed from the optimal solution to $\J$. Also recall that $\J'$ is a cost-truncated residual instance of $\J$ with respect to $\S_{pre}$.  

\begin{claim}\label{clm:spre}
     Let $S^*$ denote an optimal solution to $\MSC$ instance $\J$ whose cost is $\OPT$ and let $S'^*$ be an optimal solution to the cost-truncated residual instance $\J'$ (constructed in Line \ref{step:msc-res}) whose cost is $\OPT'$. Recall that $\J'$ is constructed with respect to $S_{pre}$ (guessed in Line \ref{step:msc-guess-sets}). Assuming $S_{pre}$ is made up of the highest cost elements of $S^*$, $\OPT' = \OPT - c(S_{pre})$
\end{claim} 

\begin{proof}
    Since $\J'$ is the cost-truncated residual instance of $\J$ with respect to $S_{pre}$, and $S_{pre}$ contains the highest cost elements of $S^*$ per our assumption, we are guaranteed that the elements of $S^* \setminus S_{pre}$ are in $N'$. We claim that $S^* \setminus S_{pre}$ is an optimal solution to $\J'$. This follows from the fact that, should there be some other solution $S'$ to $\J'$ of lower cost, then $S_{pre} \cup S'$ would be a valid solution to $\J$ whose cost is less than $\OPT$. Since no such $S'$ can exist, we can conclude that $c(S'^*)$ must be equal to $c(S^*) - c(S_{pre})$, which proves the claim. 
\end{proof}

Recall that in Stage 2 we use \Cref{thm:mle} to obtain a fractional $\vecx \in [0,1]^{N'}$ such that $\sum_j c_jx_j \leq C'$ (for $C' \geq OPT')$ and $F'_i(\vecx) \geq (1 - 1/e - \eps')b'_i$. We may assume that after using binary search to determine $\OPT$, we can set $C' := \OPT'$. For details about how to determine $\OPT$, see the last paragraph of \Cref{sec:prelim-submod}.

Next, in Stage 3 we construct a set $R$ using $\vecx$ s.t. $\E[c(R)] \leq \OPT' + O(\frac{r \ln r }{(\eps'')^2})c_{\max}$ (where $c_{\max}$ is the highest cost of an element in $N'$) and $\forall i \in [r]~f'_i(R) \geq (1-\eps')b''_i = (1-\eps'')(1-\sfrac{1}{e}-\eps')b'_i = (1-\sfrac{1}{e}-\eps)b'_i$ where the last equality holds due to our choice of $\eps'$ and $\eps''$ at the beginning of our algorithm. % - with probability $\eps'' / r$. From our choice of $\eps$, it will be the case that 

The next lemma shows that each of the $T_i$'s in Stage 4 will fix each unsatisfied constraint and have cost no more than $\OPT'$. 
\begin{lemma}\label{lem:tjs}
    For each non-empty $T_i$ (constructed in Line \ref{alg:msc-single:fix}), $f'_i(T_i) \geq (1-1/e)b'_i$ and $c(T_i) \leq \OPT'$. %, where $\OPT'$ is the optimal cost of the residual instance $\J'$.
\end{lemma}

\begin{proof}
% \PK{This proof needs to be re-written to reflect correct use of Sviridenko's result.}
    Recall from \Cref{sec:msc-pre-gr} that the problem of constructing a set $T_i$ to satisfy a constraint $f'_i$ can be reformulated as a submodular maximization problem with a single knapsack constraint. In Stage 4 of our algorithm, we do this reformulation and set the knapsack budget to $\OPT'$, which results in the following: 
\[ \max f'_i(T_i)\textnormal{ s.t. } c(T_i) \leq \OPT'\] 
    
We now claim that the optimum value for the above will be at least $b'_i$. To see this, let $T'^*_i$ denote an optimal solution to fixing the constraint $f'_i$. That is, $T'^*_i$ is the lowest cost set s.t. $f'_i(T'^*_i) \geq b'_i$. Clearly $c(T'^*_i) \leq \OPT'$, since $\J'$ may contain several other constraints that need to be simultaneously satisfied. 

Since $\OPT' \geq c(T'^*_i)$, we can conclude that the optimum value $Z$ to the above is at least $b'_i$. Thus, using the greedy algorithm from \cite{sviridenko2004note} and \Cref{lem:greedy}, we can construct a set $T_i$ that will satisfy both $c(T_i) \leq \OPT'_i$ and $f'_i(T_i) \geq (1 - \sfrac{1}{e})b'_i$. 

% Now, to see that $\OPT'_i \leq \OPT'$, observe that $\OPT'_i$ is the cost of a optimal set that will satisfy a single constraint $f'_i$, thus, if it was the case that $\OPT' < \OPT'_i$, that would imply the existence of a lower cost set that would satisfy constraint $i$, contradicting the definition of $\OPT'_i$. 
\end{proof}

Now, we show that Algorithm \ref{alg:msc-single} satisfies \Cref{thm:main-result-msc}.

\begin{proof}[Proof of \Cref{thm:main-result-msc}]

Recall that Algorithm \ref{alg:msc-single} will output $S_{pre} \cup R \cup T$. We will start by showing that the coverage requirement holds for all functions $f_i$, i.e., for all $i \in [r]$ $f_i(S_{pre} \cup R \cup T) \geq (1 - \sfrac{1}{e}-\eps)b_i$.

To prove this, we first show that $i \in [r]$, $f'_i(R \cup T) \geq  (1 - \sfrac{1}{e}-\eps)b'_i$. This follows from the fact that, after the construction of $R$ in Stage 3, we know that some constraints are sufficiently satisfied by $R$, i.e. $f'_i(R) \geq (1-\sfrac{1}{e}-\eps)b'_i$. The constraints that are not yet sufficiently satisfied are subsequently ``fixed'' in Stage 4, where, per \Cref{lem:tjs}, we construct $T_i$'s for each such constraint s.t. $f'_i(T_i) \geq (1 - \sfrac{1}{e})b'_i$. Thus, $f'_i(R \cup T) \geq (1-\sfrac{1}{e}-\eps)b'_i$ holds for all $i \in [r]$. Now, to show that the initial coverage requirement holds for each constraint $i \in [r]$, using \Cref{def:residual-msc} we have
\begin{align*}
    f_i(S_{pre} \cup R \cup T) &= f_i(R \cup T \mid S_{pre}) + f_i(S_{pre}) \\
    &= f'_i(R \cup T) + f_i(S_{pre}) \\
    & \geq (1-\sfrac{1}{e}-\eps)b'_i + f_i(S_{pre}) \\
    &= (1-\sfrac{1}{e}-\eps)(b_i -  f_i(S_{pre})) + f_i(S_{pre}) \\
    &\geq (1-\sfrac{1}{e}-\eps)b_i
\end{align*}

Next, we show that the expected cost of $S_{pre} \cup R \cup T$ is no more than $(1 + \eps)\OPT$, where $\OPT$ is the cost of an optimal solution to $\J$ (which we will denote as $S^*$). First, note that per \Cref{clm:spre}, and assuming that $S_{pre} \subseteq S^*$, $c(S_{pre}) = \OPT - \OPT'$. Next, by using \Cref{lem:additive-bound} in Stage 3, %we have that $\E[c(R)] \leq \OPT' + O(\frac{r \ln r}{(\eps'')^2})c_{\max}$. Since $\eps'' \leq \eps$ (per \Cref{clm:eps-msc}), it wil also be the case that  
we have that $\E[c(R)] \leq \OPT' +  r\lceil \frac{2 \ln (\sfrac{r}{\epsilon''}) \ln (\sfrac{1}{\epsilon''})}{(\epsilon'')^2}\rceil c_{\max}$. We can bound the $c_{\max}$ term by $\eps''\cdot c(S_{pre})$, since $c_{\max}$ will be at most the cost of any of the $r\lceil \frac{1}{\eps''}(\frac{2 \ln (\sfrac{r}{\epsilon''}) \ln (\sfrac{1}{\epsilon''})}{(\epsilon'')^2})\rceil$ elements in $S_{pre}$. Since $\eps'' \leq \eps$, we have that $\E[c(R)] \leq \OPT' + \eps \cdot c(S_{pre})$.

Finally, $\E[c(T)] = \sum_{i \in [r]} \E[c(T_i)] = \sum_{i \in [r]} Pr(T_i \neq \emptyset) \cdot c(T_i)$. Using \Cref{lem:additive-bound} and the fact that $\eps'' \leq \eps$, we can bound $Pr(T_i \neq \emptyset) \leq \sfrac{\eps''}{r} \leq \sfrac{\eps}{r}$. From \Cref{lem:tjs}, we know for any $i$ that $c(T_i) \leq \OPT'$, hence $\sum_{i \in [r]} Pr(T_i \neq \emptyset) \cdot c(T_i) \leq r \cdot \frac{\eps}{r} \cdot \OPT'  = \eps \cdot \OPT'$.

All together, we have that 
\begin{align*}
    \E[c(S_{pre} \cup R \cup T)] &= \E[c(S_{pre})] + E[c(R)] + E[c(T)] \\
    &\leq c(S_{pre}) + (\OPT' + \eps\cdot c(S_{pre})) + \eps\cdot \OPT' \\ 
    &= (1+\eps)c(S_{pre}) + (1 + \eps)\OPT' \\
    &= (1+\eps)(\OPT - \OPT') + (1 + \eps)\OPT' \\
    &= (1 + \eps)\OPT.
\end{align*}

\end{proof}

\section{Covering Coverage Functions}\label{sec:ccf}
Recall the $\CCF$ problem (\Cref{def:ccf}). In this section, we give the following guarantee for this problem.
\begin{restatable}{theorem}{thmmainresultccf}\label{thm:main-result-ccf}
    Given an instance of the $\CCF$ problem such that the underlying set cover problem has a $\beta$ approximation via the natural $\LP$, and for a deletion closed set system, for any $\epsilon > 0$, there exists a randomized polynomial time (in $\vert \Universe \vert,\, \vert \Sets \vert $) algorithm that returns a collection $\A \subseteq \Sets$ such that (i) $\A$ is a feasible solution that satisfies all the covering constraints and (ii) $\E[c(\A)] \leq \left(\frac{e}{e-1}\right)(1 + \beta)(1 + \epsilon) \OPT$.
\end{restatable}
Unlike in the previous section, here we ensure that the constraints are exactly satisfied. 
%\begin{restatable}{theorem}{thmapxfairset}\label{thm:apx-fair-set}
%     Given an instance of the $\CCF$ problem such that the underlying set cover problem has a $\beta$ approximation via the natural $\LP$ and the set system is deletion closed, for any $\epsilon > 0$, Algorithm \ref{alg:apx-ccf} returns in polynomial time (in $\vert \Universe \vert,\, \vert \Sets \vert $) a set $\A \subseteq \Sets$ such that $f_i(\A) \geq b_i$ for all $i \in [r]$ and the expected cost of this set is at most $\left(\frac{e}{e-1}\right)(1 + \beta)(1 + \epsilon)$ times the optimal cost.
%\end{restatable}
%\thmmainresultccf*
In the following, we will use $\I = (\Universe, \Sets, A, b)$ to denote an instance of the $\CCF$ problem with set system $(\Universe, \Sets)$ and constraint matrix $A$ and requirements $b_i$ for $i \in [r]$.
\subsection{Algorithm Preliminaries}
Similar to Algorithm \ref{alg:msc-single} for the $\MSC$ problem, Algorithm \ref{alg:apx-ccf} described in this section for the $\CCF$ problem follows the framework of \Cref{sec:rf-framework}. However, unlike \Cref{sec:msc}, here we want to achieve a complete satisfaction of the constraints and this requires some additional changes. Before describing the details of our algorithm in \Cref{sec:ccf-alg-overview}, we define some operations and preliminaries that we will use in our algorithm. %will be performing during the algorithm at different points of time.

\subsubsection{Operations on constraints} \label{sec:ccf-alg-operations}
At times, we will restrict our instance to a subset of elements of the universe. We define this as follows.
\begin{definition}[Restricting the Universe] Given an instance $\I = (\Universe, \Sets, A, b)$ of the $\CCF$ problem, and a subset of elements, $\Universe_{sub} \subseteq \Universe$, we define restricting the Universe to elements of $\Universe_{sub}$ as follows.
    \begin{itemize}
        \item (Update set system) For each set $S \in \Sets$, $S \gets S \cap \Universe_{sub}$ and update the Universe $\Universe \gets \Universe_{sub}$
        \item (Update constraints) For each $j \in [n]$ such that $j \notin \Universe_{sub}$, delete $j^{th}$ column of matrix A.
    \end{itemize}
\end{definition}
We assume that the underlying set system of our initial $\CCF$ instance is deletion closed. Next, our algorithm chooses sets to be a part of the solution in multiple stages. Accordingly, we update the instance to reflect this selection. This is defined as follows.
\begin{definition}[Residual Instance-$\CCF$]
    Given an initial instance $\I = (\Universe, \Sets, A, b)$ and a collection $\F \subseteq \Sets$ of sets we define a residual instance of $\I$ with respect to $\F$ as follows:
    \begin{itemize}
        \item For all $j \in \cup_{S \in \F} S$, for all $i \in [r]$, set $A_{ij} = 0$
        \item For all $i \in [r]$, set $b_i \gets b_i - f_i(\F)$.
    \end{itemize}
    In addition, for all $i \in [r]$, we denote the submodular function corresponding to the $i^{th}$ constraint by $f_{i \mid \F}(\cdot)$.
\end{definition}

\subsubsection{Greedy algorithm to fix $\CCF$ constraints}\label{sec:ccf-prelim-greedy}
Recall that our framework requires us to have the ability to satisfy a single constraint. For the $\CCF$ problem, we use the following Greedy procedure from \cite{chekuri2022covering}. Consider a universe of size $n$ with a collection of $m$ subsets of the universe. Each element $j \in [n]$ has a weight $w_j$ and each set $i \in [m]$ has a cost $c_i$. Suppose we want to pick a sub-collection of sets with maximum weighted coverage such that a budget constraint is satisfied. We can write the following $\LP$ for this problem.
\begin{equation}\label{lp:budgeted-covering}\tag{$\LP$-$\MBC$}
 \begin{array}{r@{}c@{}l}
    & \max \sum_{j \in [n]} w_jz_j \\ 
\text{s.t.}     & \sum_{i : j \in S_i} x_i \geq z_j \\ 
  &  \sum_{i \in [m]}c_ix_i  \leq   B & \\ 
  & z_j \leq 1 & \text{for all } j \in [n] \\  
   & 0 \leq x_i \leq 1  & \text{for all } i \in [m]
\end{array}
\end{equation}
This is called the $\textsc{Max-Budgeted-Cover}$ problem. One can run the standard greedy algorithm for this problem and \cite{chekuri2022covering} shows the following guarantee.
\begin{lemma}[\textcite{chekuri2022covering}]\label{lem:greedy-ccf-single}
    Let $Z$ be the optimum value of (\ref{lp:budgeted-covering}) for a given instance of $\textsc{Max-Budgeted-Cover}$ with budget $B$. Suppose Greedy Algorithm is run till the total cost of sets is equal to or exceeds $B$ for the first time. Then the weight of elements covered by greedy is at least $(1 - \frac{1}{e})Z$.
\end{lemma}
\subsection{Algorithm Overview}\label{sec:ccf-alg-overview}
\paragraph{Stage 1: Guess the high cost elements.} Following the framework, we guess $(\frac{r}{\lipc} \ln \frac{1}{\epsilon} + r)$ high cost elements from some fixed optimal solution (Line \ref{step:guess-sets}). The number $\lipc$ is chosen to be $\frac{\epsilon^2}{2\ln \sfrac{r}{\epsilon}}$. The choice of this number will be clear in Analysis (Section \ref{sec:ccf-analysis}). After this, we construct a residual instance to reflect this selection and work with the residual instance for remaining stages.

\paragraph{Stage 2: Construct Fractional Solution $\vecx$ using the LP Relaxation.} Recall that the objective in $\CCF$ is to solve the integer program \ref{ip:fair-covering}. To obtain a fractional solution, in Line \ref{step:lp-sol} we solve the linear relaxation of this program \ref{lp:fair-covering} for the residual instance obtained after Stage 1 and obtain a solution $(\vecx^*, \vecz^*)$.
\paragraph{Stage 2$'$: Obtain a slack in covering.} 
%By solving \ref{lp:fair-covering}, we get a fractional solution $(\vecx^*, \vecz^*)$ that satisfies $\sum_{j \in [n]}A_{ij}z^*_j \geq b_i$. To use \Cref{lem:additive-bound}, we need a guarantee with respect to the multilinear extension of the function. However, the guarantees we have for $\vecx$ are with respect to a continuous extension known as the \emph{concave closure}, denoted by $f_i^*$. Luckily, it is known that the value of multilinear extension is point-wise at least $(1-\frac{1}{e})$ times the value of $f^*$ \cite{chekuri2011submodular}. %\footnote{This is called the correlation gap \cite{chekuri2011submodular}}. 
%Therefore, we get that under multilinear extension $F_i$, the fractional solution satisfies $F_i(\vecx^*, \vecz^*) \geq (1 - \frac{1}{e})b_i$. 

%In $\CCF$ problem, we want to cover the constraints fully. Since we lose value in different parts of rounding, we need to obtain a slack in the covering and for this purpose, we want to scale the solution by $\scaleccf.$ The choice of this scaling factor will be clear in Analysis (Section \ref{sec:ccf-analysis}). However, if all the $z^*_j$ are scaled by $\scaleccf$, then we would get $\sum_{j \in [n]}A_{ij}z^*_j \geq \scaleccf b_i$. This would allow us to apply \Cref{lem:additive-bound} in the next Stage and obtain full coverage. However, scaling $\vecx^*$ (or $\vecz^*$) by $\scaleccf$ directly will not always work since both $x_i$ and $z_j$ are bounded above by $1$ (and scaling may violate this upper bound). Therefore, 
We divide the elements into two categories: The \emph{heavy} elements ($H_e$) are those for which $z_j^* \geq (1-\frac{1}{e})(1 - \epsilon)$, while the \emph{shallow} elements are the remaining non-heavy elements. 

We will first cover the heavy elements completely (in Line \ref{step:heavy-elts}). This is done by restricting the universe to $H_e$, using $(\vecx^*, \vecz^*)$ restricted to $H_e$ as a solution for canonical $\LP$ of the canonical set cover problem and then using the $\beta$-approximation promised in the problem to obtain an integral covering. Here we use the fact that the instance is deletion closed. We then scale the shallow elements by $\scaleccf$ to obtain a corresponding slack in covering. 
\paragraph{Stage 3: Round fractional solution.} We take the scaled fractional covering of shallow elements and round it using a call to \Cref{lem:additive-bound} (Line \ref{step:scale-and-round}).

\paragraph{Stage 4: Greedy fixing of the solution.} For any constraint $i \in [r]$ that is not satisfied, we will fix it via a greedy algorithm. Following the greedy fix for $\textsc{Max-Budgeted-Cover}$ in \ref{sec:ccf-prelim-greedy}, this works as follows: From the collection of sets not picked, pick the sets greedily in order of decreasing bang-per-buck (ratio of marginal value of the set and its cost). 
%To use this in our algorithm, first note that any element that is covered heavily has already been picked in $\Sets_{he}$. Therefore, for any constraint $i \in [r]$ given by $\sum_{j \in [n]}A_{ij}z_j \geq b_i$, if it is not yet satisfied, we need to cover some more shallow covered elements. Now, when we restrict the universe to shallow covered elements the scaled solution $\scaleccf (\vecx^*, \vecz^*)$ is a feasible solution to the \ref{lp:budgeted-covering} for the $i^{th}$ constraint with a slack of $\scaleccf$ in the coverage. Therefore, applying Lemma \ref{lem:greedy-ccf-single}, at the cost of $\sum_{i \in [m]}c_ix^*_i + c_{\max}$ where $c_{\max}$ is the highest cost set in the collection left, we get a full coverage of shallow covered elements. Together with $\Sets_{he}$, we get a full coverage of the functions. The details of analysis are in Section \ref{sec:ccf-analysis}.
\begin{algorithm}[h!]
\caption{Pseudocode for Covering Coverage Functions}
  \label{alg:apx-ccf}
\DontPrintSemicolon
  \SetKwFunction{Define}{Define}
  \SetKwInOut{Input}{Input}\SetKwInOut{Output}{Output}
  \Input{An instance of $\CCF$ problem denoted by $\I = (\Universe, \Sets, A, b)$}
  \Output{A subset, $\A \subseteq \Sets$ of the sets satisfying the claims in Theorem \ref{thm:main-result-ccf}}
  \BlankLine
  With $\lipc = \frac{ \epsilon^2}{2\ln \sfrac{r}{\epsilon}}$ construct $\Sets_{pre}$ by guessing 
  the $r \cdot \lceil \frac{1}{\lipc}\ln \frac{1}{\epsilon}\rceil + r$ highest cost sets from a fixed optimal solution of $\I$. \tcp*[r]{Stage 1}\label{step:guess-sets} 
  Eliminate all sets with cost higher than any set in $\Sets_{pre}$ and then create $\I'$ as residual instance of $\I$ with respect to $\Sets_{pre}$\label{step:prune-large-cost}\;
  Compute $(\vecx^*, \vecz^*)$ as the optimal solution to linear program relaxation for the instance $\I'$, given in \ref{lp:fair-covering} \tcp*[r]{Stage 2}\label{step:lp-sol}
  Define $H_e \coloneqq \{j \mid z^*_j \geq (1 - \frac{1}{e}) \cdot (1 - \epsilon)\}$. Restrict the $\Universe$ to $H_e$, use $\scaleccf (\vecx^*, \vecz^*)$ as solution for the canonical set covering $\LP$ and use the corresponding algorithm to cover all elements in $H_e$. Call the sets selected here $\Sets_{he}$ \tcp*[r]{Stage 2'} \label{step:heavy-elts}
   Restrict $\Universe$ to $\Universe \setminus H_e$ and use $\scaleccf \cdot (\vecx^*, \vecz^*)$ as solution to fairness constrained \ref{lp:fair-covering} for the instance $\I'$. Round this using Lemma \ref{lem:additive-bound} and call the sets chosen here $\Sets_{sh}$ \tcp*[r]{Stage 3} 
 \label{step:scale-and-round}
  For $i \in [r]$ if $f_i$ is not satisfied, create $\G_i$ by choosing greedily (with respect to marginal value) sets from $\Sets$ till it is satisfied. \tcp*[r]{Stage 4} \label{step:greedy-fix}
  \Return{$\Sets_{pre} \cup {\Sets}_{he} \cup {\Sets}_{sh} \cup \big( \bigcup_{i \in [r]}\G_i \big)$}
\end{algorithm}
\subsection{Analysis}\label{sec:ccf-analysis}
First note that since in Line $\ref{step:greedy-fix}$ we fix any unsatisfied constraint, the collection of sets returned by Algorithm \ref{alg:apx-ccf} clearly covers all constraints fully. Therefore, we now focus on establishing a bound on cost of the sets selected and the time complexity of the algorithm.
Before proceeding, we introduce some notation. Let the cost of integral optimal solution be $\OPT$, the cost of highest cost $ (r \lceil \frac{1}{\lipc}\ln(\frac{1}{\epsilon})\rceil + r)$ sets from integral optimal be $\OPT_\h$ and the cost of remaining sets be $\OPT_{\l}$. For the $\LP$ solution obtained in Line \ref{step:lp-sol}, let the cost be $\OPT$-$\LP$. Finally, for any collection of sets, $\A$, we use $c(\A)$ to denote the total cost of these sets.

\begin{claim}\label{clm:ccf-high-sets}
    $c(\Sets_{pre}) = \OPT_\h$
\end{claim}
\begin{proof}
    This follows directly since in Line \ref{step:guess-sets} we guess sets from the optimal solution. As a result, these sets  will not contribute any more than what they do in the optimal.
\end{proof}
\begin{claim}\label{clm:ccf-lp-cost}
    $\OPT$-$\LP \leq \OPT_{\l}$
\end{claim}
\begin{proof}
    We solve \ref{lp:fair-covering} for the instance $\I'$. This instance is constructed by (1) Selecting some highest cost sets from optimal (2) Removing all sets with cost higher than any set picked. Since the higher cost sets cannot be part of the optimal solution, the sets constituting $\OPT_{\l}$ form a feasible solution for the $\LP$-$\CCF$ of $\I'$ and this proves the claim. 
\end{proof}
\begin{lemma}
    $c(\Sets_{he}) \leq \beta \cdot (\frac{e}{e-1}) \cdot (\frac{1}{1-\epsilon})\cdot \OPT_\l$
\end{lemma}
\begin{proof}
Consider the fractional solution where $\Tilde{\vecx}^* = (\frac{e}{e-1})\cdot (\frac{1}{1-\epsilon}) \vecx^*$ and $\Tilde{\vecz}^* = \min \{ (\frac{e}{e-1})\cdot (\frac{1}{1-\epsilon}) \vecz^*, 1\}$ where the minimum is element wise. Restrict the universe to elements in $H_e$ (defined in Algorithm line \ref{step:heavy-elts}). By definition, for all $j \in H_e$, $z_j^* \geq (1- \frac{1}{e})(1-\epsilon)$ therefore, $\Tilde{z}_j^* = 1$. This implies $(\Tilde{\vecx}^*, \Tilde{\vecz}^*)$ is a feasible solution to the natural Set Cover $\LP$ for covering $H_e$. Let the optimal value of this natural $\LP$ be $\OPT$-$\LP_{he}$. Since we assume that the underlying set cover family is deletion closed and has a $\beta$-approximation via its natural $\LP$, we can use this underlying algorithm to obtain $\Sets_{he}$ whose cost is a $\beta$-approximation to $\OPT$-$\LP_{he}$.  Therefore,
    \begin{align*}
        c(\Sets_{he}) \leq \beta \cdot \OPT\text{-}\LP_{he} \leq \beta \cdot c(\Tilde{\vecx}^*, \Tilde{\vecz}^*) &\leq \beta \cdot \left(\frac{e}{e-1}\right) \cdot \left(\frac{1}{1-\epsilon}\right)\cdot \OPT\text{-}\LP \\
        &\leq \beta \scaleccf \OPT_\l
    \end{align*}
   where $c(\Tilde{\vecx}^*, \Tilde{\vecz}^*) = \sum_{i \in [m]} c_i\Tilde{x}^*_i$. This proves the lemma.
\end{proof}
Before we move on with the proof of the theorem, we will prove a technical claim regarding the constraints in the $\CCF$ problem. Recall that $f_i$ is the submodular function form of the $i^{th}$ constraint and $\sum_{j \in [n]}A_{ij}z_j$ is the constraint from the matrix form.
\begin{claim}\label{clm:ccf-extensions}
    For all $i \in [r]$, given any fractional point $(\vecx, \vecz)$ that is a feasible solution to \ref{lp:fair-covering}, we have $F_i(\vecx) \geq (1 - \frac{1}{e})\sum_{j \in [n]}A_{ij}z_j$, where $F_i$ is the multilinear extension of $f_i$. 
\end{claim}
\begin{proof}
    Since $(\vecx, \vecz)$ is a fractional point, $\sum_{j \in [n]}A_{ij}z_j = \sum_{j \in [n]} A_{ij} \min\{1, \sum_{i: j \in S_i} x_i\}$ can be viewed as a continuous extension of the corresponding submodular function $f_i$ with respect to the variablex $\vecx$. In fact, for a (weighted) coverage function $f$, this extension is called the $f^*$ extension. The value of $f^*$ is point-wise at least as much as $F$ and the gap between them is known to be at most $1-\sfrac{1}{e}$ (see \cite{CCPV07,Vondrak-thesis}). Therefore, for the $i^{th}$ constraint, $F_i(\vecx) \ge (1-1/e) \sum_{j \in [n]}A_{ij}z_j$, 
    which proves the claim.
\end{proof}
\begin{lemma}\label{lem:ccf-greedy-sets-bound}
    $\E[c(\Sets_{sh} \cup \bigcup_{i \in [r]} \G_i)] \leq \OPT_\h + \scaleccf \cdot \OPT_\l + \epsilon \OPT$.
\end{lemma}
\begin{proof}
    Consider first the sets $\Sets_{sh}$. These are the sets returned by Lemma \ref{lem:additive-bound}. In the algorithm, we invoke Lemma \ref{lem:additive-bound} to round the fractional solution $\scaleccf (\vecx^*, \vecz^*)$ restricted to $\Universe \setminus H_e$. Since the fractional solution is obtained for the residual instance after Stage 1, the cost of this fractional solution is $\scaleccf \sum_{i \in \Sets \setminus \Sets_{pre}} c_ix^*_i = \scaleccf \OPT\text{-}\LP$. Therefore defining $c_{max} = \max_{S \in Sets \setminus \Sets_{pre}} c_S$ and using Lemma \ref{lem:additive-bound}, the cost of $\Sets_{sh}$ is bounded as follows
    \begin{align*}
        c(\Sets_{sh}) &\leq \left( \frac{r}{\lipc} \ln \frac{1}{\epsilon}\right) c_{\max} + \scaleccf \OPT\text{-}\LP \\
        &\leq \left(\frac{r}{\lipc} \ln \frac{1}{\epsilon}\right) c_{\max} + \scaleccf \OPT_\l \\
    \end{align*}
    Here the second inequality follows from Claim \ref{clm:ccf-lp-cost}.

    Now, we want to bound the expected cost of $\cup_{i \in [r]} \G_i$. To do this, first we will bound the probability that the algorithm applies a greedy fix. Consider any constraint $i \in [r]$ given by $\sum_{j \in [n]}A_{ij}z_j$. If the constraint is satisfied by the elements covered in $\Sets_{pre} \cup \Sets_{he} \cup \Sets_{sh}$ the greedy fix will not apply. Therefore, suppose the constraint is not satisfied after this point. We will bound the probability of this event. Suppose the residual requirement for the $i^{th}$ constraint after Line \ref{step:guess-sets} (guessing $\Sets_{pre}$) is $b_i'$ and the residual constraint matrix is given by $(A'_{ij})_{i \in [n], j \in [m]}$. Then, with the $\LP$ solution $(\vecx^*, \vecz^*)$, we have $\sum_{j \in [n]}A'_{ij}z^*_j \geq b_i'$. After this, we cover elements in $H_e$ using $\Sets_{he}$. Suppose the residual requirement after covering elements in $H_e$ is $b_i''$. Then we have $\sum_{j \in [n] \setminus H_e} A'_{ij}z_j^* \geq b_i''$. After this, we scale the solution by $\scaleccf$. Let this scaled solution be $(\Tilde{\vecx}^*, \Tilde{\vecz}^*)$. We then have $\sum_{j \in [n] \setminus H_e} A'_{ij}\Tilde{z}_j^* \geq \scaleccf b''_i$. To use the guarantees of Lemma \ref{lem:additive-bound}, we need to work with multilinear extension. Using \Cref{clm:ccf-extensions}, $F'_i(\Tilde{\vecx}^*) \geq (1-\frac{1}{e}) \cdot \scaleccf b''_i = (\frac{1}{1-\epsilon})b''_i$. Therefore, using the guarantee of Lemma \ref{lem:additive-bound}, we get that the returned solution satisfies $f'_i(\Sets_{sh}) \geq b''_i$ with probability at least $1 - \frac{\epsilon}{r}$. Therefore, with probability at least $1 - \frac{\epsilon}{r}$ we have,
    \begin{equation}\label{eqn:ccf-total-cover}
    \begin{aligned}
        f_i(\Sets_{pre} \cup \Sets_{he} \cup \Sets_{sh}) &= f_i(\Sets_{pre}) + f_i(\Sets_{he} \cup \Sets_{sh} \mid \Sets_{pre}) \\ 
        &= f_i(\Sets_{pre}) + f'_i(\Sets_{he} \cup \Sets_{sh}) \\
        &= (b_i - b'_i) + \sum_{j \in H_e} A'_{ij} + \sum_{j \in (\cup_{S \in \Sets_{sh}} S)} A'_{ij} \\
        &= (b_i - b'_i) + (b'_i - b''_i) + (b''_i) \\
        &= b_i
    \end{aligned}
    \end{equation}
    This implies, the probability that greedy fix is applied for a constraint is at most $\frac{\epsilon}{r}$. Now to bound the total expected cost, we will bound the cost of each $\G_i$ in the following claim.

    \begin{claim}
        For all $i \in [r]$, $c(\G_i) \leq \scaleccf \OPT\text{-}\LP + c_{max}$. 
    \end{claim}
    \begin{proof}
     Recall from Algorithm \ref{alg:apx-ccf}, $\G_i$ is selected by picking greedily until the constraint is satisfied. To bound the cost of this, we have to connect the value of sets picked greedily and the cost of these sets. \Cref{lem:greedy-ccf-single} does exactly this for us. In particular, note that $(\Tilde{\vecx}^*, \Tilde{\vecz}^*)$ is a valid solution to \ref{lp:budgeted-covering} for any fixed constraint (say, $i \in [r]$). The cost of this solution is $\scaleccf \OPT\text{-}\LP$ and $\sum_{j \in [n] \setminus H_e} A'_{ij}z_j^* \geq b_i''$. Therefore, using Lemma \ref{lem:greedy-ccf-single}, we get that we can select sets $\Sets_{sh}$ such that $c(\Sets_{sh}) \leq \scaleccf \OPT\text{-}\LP + c_{\max}$ where $c_{\max}$ is the highest cost set in $\Sets \setminus \Sets_{pre}$, and  $f'_i(\Sets_{sh}) \geq (\frac{1}{1-\epsilon}) b_i'' \geq b_i''$. Following the same argument as in $\Cref{eqn:ccf-total-cover}$, we get that this satisfies the covering constraints with probability 1.
    \end{proof}
     We can now continue with our proof. The expected cost of $\cup_{i \in [r]} \G_i$ is bounded as follows:
    \begin{align*}
        \E[c(\cup_{i \in [r]} \G_i)] &\leq r \cdot \left(\frac{\epsilon}{r} \cdot \scaleccf \OPT\text{-}\LP + c_{\max}\right) \\
        &\leq \epsilon' \OPT + r c_{\max}
    \end{align*}
   where $\epsilon' = \left(\frac{\epsilon}{1-\epsilon} \cdot \frac{e}{e-1}\right)$. Together, we get
    \begin{align*}
        \E[c(\Sets_{sh} \cup \cup_{i \in [r]}\G_i )] &\leq \left(\frac{r}{\lipc} \ln \frac{1}{\epsilon}\right) c_{\max} + \scaleccf \OPT_\l + \epsilon' \OPT + r c_{\max} \\
        &\leq c(\Sets_{pre}) + \scaleccf \OPT_\l + \epsilon' \OPT \\
        &\leq \OPT_\h + \scaleccf \OPT_\l + \epsilon' \OPT .
    \end{align*}
    Since $\epsilon < \epsilon'$, to get the claim for some $\epsilon^*$, we first fix $\epsilon' = \epsilon^*$ then choose $\epsilon$ by solving for $\epsilon$ in $\epsilon' = \left(\frac{\epsilon}{1-\epsilon} \cdot \frac{e}{e-1}\right)$. For any constant $\epsilon' > 0$, $\epsilon$ is also a constant and $\epsilon > 0$. Therefore, the claim is proved.
\end{proof}

\begin{proof}[Proof of Theorem \ref{thm:main-result-ccf}]
    Consider the total cost of the solution.
    \begin{align*}
        \E[c(\Sets_{pre} \cup \Sets_{he} \cup \Sets_{sh} \cup \cup_{i \in [r]}\G_i)] &\leq \E[c(\Sets_{pre})] + \E[c(\Sets_{he})] + \E[c(\Sets_{sh} \cup_{i \in [r]}\G_i)] \\
        &\leq \OPT_\h + \beta \scaleccf \OPT_\l + \OPT_\h \\ &\qquad \qquad \qquad+ \scaleccf \OPT_\l + \epsilon \OPT \\
        &\leq 2 \OPT_\h + \scaleccf(\beta+1)\OPT_\l + \epsilon \OPT \\
        &\leq \left(\frac{e}{e-1}\right)(1+\epsilon')(1+\beta) \OPT.
    \end{align*}
    This proves the cost of the solution. Finally, to see that the algorithm runs in polynomial time, note that the costliest operation is selection of $\Sets_{pre}$. We are choosing $\left(\frac{r}{\lipc}\ln \frac{1}{\epsilon} + r\right)$ sets of highest cost from the optimal for $\lipc = \frac{\epsilon^2}{2 \ln \sfrac{r}{\epsilon}}$. This will need an enumeration of at most $m^{\frac{2r}{\epsilon^2}\ln\frac{1}{\epsilon} \ln \frac{r}{\epsilon}} = O(m^{\frac{r}{\epsilon^2} \ln^2 \frac{r}{\epsilon}})$ different collections of $\Sets_{pre}$. Since both $r$ and $\lipc$ are constants, this takes polynomial time. All other steps of the algorithm clearly run in polynomial time. This completes the proof.
\end{proof}

\section{Applications of $\CCF$} \label{sec:ccf-apps}
In this section, we show via two applications how the $\CCF$ framework can be used in different applications. For the first application in Section \ref{sec:flccc}, we cannot directly apply the $\CCF$ framework. However, with some technical modifications, show how the main ideas of the framework can still be implemented to obtain a good approximate result. Our second application in Section \ref{sec:radii} follows from a more direct application of the framework.
\subsection{Facility Location with Multiple Outliers}\label{sec:flccc}
We consider a generalization of the classical facility location problem, known as Facility Location with Multiple Outliers ($\FLMO$), in which clients belong to color classes, and the objective is to cover a required number of clients from each class. Specifically, we are given a set of facilities $F$, a set of clients $C$, where each client belongs to at least one of $r$ color classes $C_k$ (i.e. $C = \cup_{k\in[r]} C_k$), and a metric space $(F \cup C, d)$. Each facility $i \in F$ has an associated non-negative opening cost $f_i$, and each client $j \in C$ can be served by assigning it to an open facility $i$, incurring connection cost $d(i,j)$. For each color class $C_k$, we are given a required demand $b_k$ (where $1 \leq b_k \leq |C_k|$). The goal is to open facilities and assign clients to facilities such that at least $b_k$ clients from $C_k$ are serviced by a facility, and where the total facility and connection costs of opened facilities and their assigned clients is minimized. 

$\FLMO$ can be naturally expressed as an instance of the $\CCF$ problem over an exponentially large set system, where the universe of elements is $C$ and where each set corresponds to a \emph{star} $(i, S)$ defined by a facility $i \in F$ and a subset of clients $S \subseteq C$ assigned to it. The goal is to select a collection of such stars to cover at least $b_k$ clients from each color class $C_k$, while minimizing the total cost, where the cost of a star is given by $c(i, S) := f_i + \sum_{j \in S} d(i, j)$. We can capture this problem via the following exponentially sized CCF-esque IP. 

\begin{equation}\label{ip:flccc}\tag{$\IP$-$\FLMO$}
\begin{array}{llll}
    & \min & \displaystyle\sum_{i \in F} \sum_{S \subseteq C} c(i, S) \cdot x(i, S) \\
    & \text{s.t.} & \displaystyle\sum_{i \in F} \sum_{\substack{S \subseteq C \\ j \in S}} x(i, S) \geq z_j & \forall j \in C \\
    & & \displaystyle\sum_{j \in C_k} z_j \geq b_k & \forall k \in [r] \\
    & & x(i, S),~ z_j \in \{0,1\} & \forall i \in F, S \subseteq C, j \in C
\end{array}
\end{equation}

We observe that this formulation does not require the assumption that $d$ is a metric. This will be useful in solving restricted version of a related LP relaxation that we will see later.
Since there are exponentially many stars to be considered (resulting in an IP with an exponential number of variables), we cannot directly apply the techniques from our algorithm for $\CCF$. 

We apply a refined version of the $\CCF$ framework that avoids explicitly enumerating all possible stars. The most notable changes are to Stages 1 and 2. In our new implementation of Stage 1, instead of guessing complete high-cost sets from an optimum solution, we guess certain high-cost components of the optimal solution. After guessing, we appropriately construct a (still potentially exponentially large) residual instance. To obtain a fractional solution for this without explicitly describing the full LP, we use the dual. Specifically, the separation oracle for dual can be run in polynomial time and allows us to solve the dual via the ellipsoid method, while identifying only a polynomial number of tight constraints. These then correspond to a polynomial-sized set of stars from the primal that form the support of an optimal solution. We therefore restrict the original primal LP to just this support and treat this reduced instance as our effective set system for the remaining stages. To ensure that the aforementioned steps can be executed in polynomial time, we must first implement a pre-processing stage to scale and discretize all distances and facility opening costs to ensure that star costs become are integer-valued and polynomially bounded. %; this will be vital to ensure both the guessing and LP-solving stages of our framework can be done in polynomial time.

We now describe each stage of our refined $\CCF$-based algorithm for $\FLMO$, beginning with a pre-processing step that enables efficient guessing and separation oracle.

\subparagraph*{Stage 0: Scaling and pruning facility and connection costs.}
We begin with a pre-processing step that allows us to perform guessing and LP-solving steps efficiently. As is standard, we can assume that the value of the optimal solution $\OPT$ is guessed up to a $(1 + o(1))$ factor (we overload $\OPT$ to denote both the true and guessed optimal cost). This introduces only a polynomial overhead in runtime and at most a $(1 + o(1))$ loss in the approximation ratio. In addition to guessing $\OPT$, we define a polynomial scaling factor $B:= n^3$ which will be used to scale the relevant facility and connection costs in our instance.

Using the guessed $\OPT$ and $B$, we can now eliminate any unnecessary or negligible facility and connection costs. First, we disallow (i.e. remove from the instance) any facilities $i$ for which $f_i > \OPT$ since we know they will never be selected in the optimal solution. We also disallow any connections between facilities $i$ and clients $j$ for which $d(i,j) > \OPT$. Next, for any facility $i$ with $f_i\leq \sfrac{\OPT}{B}$, we define its scaled cost $\bar{f}_i := 0$; for any client-facility pair $(i,j)$ where $d(i,j) \leq \sfrac{\OPT}{B}$, we define the scaled connection cost $\bar{d}(i,j) := 0$. This is done because these costs are too small to meaningfully contribute to the cost of the final solution, and hence can be safely ignored. 

For the remaining facility and connection costs, we scale them as follows: 
\[
\bar{f}_i := \left\lceil \sfrac{B}{\OPT} \cdot f_i \right\rceil, \qquad
\bar{d}(i,j) := \left\lceil \sfrac{B}{\OPT} \cdot d(i,j) \right\rceil.
\]
This will guarantee that all scaled facility costs will be integer-values in $[0,B]$. Thus, the scaled cost of each star $\bar{c}(i,S) := \bar{f}_i + \sum_{j \in S} \bar{d}(i,j)$ will be some polynomially bounded integer. It is important to note that the discretized distances $\bar{d}(i,j)$ may no longer form a metric (since they may violate triangle inequality). However, this is fine our $\CCF$ framework did not require set costs to satisfy such properties. The metric property for distances will only be required in Stage 2', when we must solve the canonical facility location problem to cover heavy clients. For that stage, we will show that we can viably revert back to using the original metric $d$.

This pre-processing step will incur at most an additive $O(\sfrac{\OPT}{B})$ term per cost term, and thus will result in a $1 + o(1)$ multiplicative increase in the total cost of the solution, but this can be absorbed into our final approximation factor. In the rest of the section we assume that we have guessed $\OPT$ correctly to within a $(1 + o(1))$ factor.

\subparagraph*{Stage 1: Guessing high-cost components.} This stage consits of two parts: (1) Guess some high cost components (2) Create a residual instance accounting for the guess.

\noindent \textbf{Guess high-cost components.} In the $\CCF$ framework, we needed to guess the $L := \left(\frac{r}{\lipc} \ln \sfrac{1}{\epsilon} + r\right)$ most expensive stars from the optimal solution. This is because our rounding procedure in \Cref{lem:additive-bound} and the greedy fix step later choose these many extra sets to satisfy the constraints and we account for these sets via the initial guesses. 
%This is done by ensuring that the remaining sets in our residual instance satisfy a Lipschitz condition. In particular, we must ensure that the marginal contribution of any set is bounded by $\ell\cdot B$ where $B$ denotes the fractional cost of the residual solution. 
For the $\FLMO$ setting explicitly guessing a a star $(i, S)$ would require exponential time. To circumvent this we instead guess partial information about many stars. However, we must now guess more information due to this partial knowledge. 

Specifically, we guess tuples of the form $(i_h, S_h, g_h)$ for $h \in [T]$ with $T = \lceil L/\epsilon \rceil$, where $i_h$ is a facility from one of the $T$ highest cost stars, $S_h \subseteq C$ is a set of the $L$ farthest clients assigned to $i_h$ in the optimal solution, and $g_h$ is the total cost of the full optimal star at facility $i_h$, i.e., $\bar{c}(i_h, S^*_h)$, where \( S^*_h \) denotes the full client set served by \( i_h \). We note that if there are fewer than $T$ stars in the optimum solution, the problem becomes simpler; in that case we would have guessed all the optimum's facilities. This will be clear after a description of the remaining analysis and we give a remark at the end of Stage $3$ as to how to deal with that case. Let $\OPT_{guess}$ be the cost of the guessed portion of the solution.

%To make sure this partial guess is enough to recover the same Lipschitz-like condition needed in the proof of Lemma~\ref{lem:additive-bound}, we must ensure that the clients in $S_h$ account for most of the cost of the optimal star. 

%To ensure the effectiveness of our guess, we must have $t$ be some constant that will be large enough to guarantee that the clients in $S_h$ account for most of the cost of the optimal star. This can be accomplished by guessing at least $\frac{1}{\ell} \log \frac{1}{\epsilon}$ of the highest-cost (i.e. furthest) clients from each of the highest $L$ cost stars. This is because, in the original $\CCF$ analysis, ensuring that we can satisfy constraints with high probability is reliant on ensuring that the remaining sets in our residual instance satisfy a Lipschitz condition. In particular, we must ensure that the marginal contribution of any set is bounded by $\ell\cdot B$ where $B$ denotes the fractional cost of the residual solution. For the $\FLMO$ setting additional $r$ clients for each $S_h$. All together, $S_h$ must contain the farthest $t := \frac{1}{\ell} \log \frac{1}{\epsilon} + r$ clients in the optimal star around $i_h$ in order for our the remainder of our analysis to go through.  % Recall that $\ell \in \Theta(\frac{\eps^2}{\log(r/\eps)})$, thus
%\[ t \in \Theta\left(\frac{1}{\ell}\log(1/\eps)\right) = \Theta\left(\frac{\log(r/\eps)}{\eps^2}\cdot\log(1/\eps)\right) = \Theta\left(\frac{1}{\eps} \log(r/\eps)\right).\] 

\noindent \textbf{Create residual instance.} Let $S_{pre} := \{(i_h, S_h, g_h)\}_{h \in [L]}$ denote the collection of guessed partial stars. We define $F_{pre} := \{ i_h \mid (i_h, S_h, g_h) \in S_{pre} \}$ as the set of guessed high-cost facilities, and $C_{pre} := \bigcup_{h \in [L]} S_h$ as the set of clients guessed to be served by those stars. Let $G := \min_{h \in [L]} g_h$ be the smallest guessed star cost across all tuples.

Now, to describe the residual instance with respect to these guessed components, first update the color demands $b_k$ to $b'_k := b_k - |C_k \cap C_{pre}|$. Next, restrict attention to clients in $C \setminus C_{pre}$ and define for each facility $i \in F$ a restricted family of allowable stars, denoted by $\mathcal{S}_i$, and updated star costs, denoted by $\bar{c}'$, as follows:
\begin{itemize}[leftmargin=*]
    \item For $i \in F_{pre}$, we allow only singleton stars $(i, \{j\})$ with $j \in C \setminus C_{pre}$ and $\bar{d}(i,j) \leq \min_{j' \in S_h} \bar{d}(i,j')$ (for the corresponding tuple $(i, S_h, g_h) \in S_{pre}$). For these stars, we define the updated cost as $\bar{c}'(i, \{j\}) := \bar{d}(i,j)$, since after guessing $S_{pre}$, we account $\bar{f}_i$ to open facility $i$ and the connection costs for the clients in $S_h$.
    
    \item For $i \notin F_{pre}$, we allow any star $(i, S)$ with $S \subseteq C \setminus C_{pre}$ and total cost at most $G$. For these stars, the cost remains as $\bar{c}'(i, S) := \bar{f}_i + \sum_{j \in S} \bar{d}(i,j)$.
\end{itemize}
% \begin{equation}\label{lp:flccc}\tag{$\LP$-FL-CCC}
% \begin{array}{llll}
%     & \min & \displaystyle\sum_{i \in F} \sum_{(i,S) \in \mathcal{S}_i} \bar{c}(i, S) \cdot x(i, S) \\
%     & \text{s.t.} & \displaystyle\sum_{i \in F} \sum_{\substack{(i,S) \in \mathcal{S}_i: \\ j \in S}} x(i, S) \geq z_j & \forall j \in C \setminus C_{pre} \\
%     & & \displaystyle\sum_{j \in C_k \setminus C_{pre}} z_j \geq b_k' & \forall k \in [r] \\
%     & & z_j \leq 1 & \forall j \in C \setminus C_{pre} \\
%     & & x(i, S), z_j \geq 0 & \forall i \in F, (i, S) \in \mathcal{S}_i, \forall j \in C \setminus C_{pre}
% \end{array}
% \end{equation}

% Note that despite culling some amount of stars after guessing, the number of residual stars and thereby the LP may still be exponential in size. 

\begin{equation}\tag{$\LP$-$\FLMO$ Primal \& Dual}\label{lp:flccc+dual}
\begin{minipage}{0.48\textwidth}
\textbf{Restricted Primal ($\LP$-$\FLMO$)}\\[0.5em]
$\displaystyle \min\ \sum_{i \in F} \sum_{(i,S) \in \mathcal{S}_i} \bar{c}'(i, S)\cdot x(i, S)$ \\[1em]
s.t.\\
$\displaystyle \sum_{i \in F} \sum_{\substack{(i,S) \in \mathcal{S}_i:\\ j \in S}} x(i, S) \geq z_j \quad \forall j \in C \setminus C_{pre}$ \\[0.5em]
$\displaystyle \sum_{j \in C_k \setminus C_{pre}} z_j \geq b_k' \quad \forall k \in [r]$ \\[0.5em]
$z_j \leq 1 \quad \forall j \in C \setminus C_{pre}$ \\[0.5em]
 $0 \leq x(i, S) \leq 1\quad \forall i \in F,\ (i, S) \in \mathcal{S}_i$ \\
 $0 \leq z_j \quad \forall j \in C \setminus C_{pre}$
\end{minipage}
\hfill
\begin{minipage}{0.48\textwidth}
\textbf{Dual}\\[0.5em]
$\displaystyle \max\ \sum_{k=1}^r b_k' \cdot \beta_k - \sum_{j \in C \setminus C_{pre}} \gamma_j$ \\[1em]
s.t.\\
$\displaystyle \sum_{j \in S} \alpha_j \leq \bar{c}'(i, S) \quad \forall i \in F,\ (i, S) \in \mathcal{S}_i$ \\[0.5em]
$\alpha_j - \gamma_j \leq \beta_k \quad \forall j \in C_k \setminus C_{pre},\ \forall k \in [r]$ \\[0.5em]
$\alpha_j, \beta_k, \gamma_j \geq 0 \quad \forall j \in C \setminus C_{pre}, ~k \in [r]$
\end{minipage}
\end{equation}

\subparagraph*{Stage 2: Constructing a fractional solution via the dual.}
We now solve an LP relaxation for the residual instance defined in Stage 1. This may still contain exponentially many variables. To solve this LP efficiently, we apply the ellipsoid method to its dual (see \ref{lp:flccc+dual}), which has polynomially many variables but exponentially many constraints. This requires a polynomial-time separation oracle which, given a candidate dual solution $(\alpha, \beta, \gamma)$, either certifies feasibility or returns a violated constraint. That is, it identifies residual star $(i,S) \in \mathcal{S}_i$ such that $\sum_{j \in S} \alpha_j > \bar{c}'(i,S)$.
Again, since facility costs $\bar{f}_i$ and distances $\bar{d}(i,j)$ are integral and polynomially bounded (from pre-processing in Stage 0), we can design such an oracle. We formalize this as the following lemma.

\begin{lemma}\label{lem:sep-oracle}
Assuming that all distances $\bar{d}(i,j)$ and facility costs $f_i$ are integral and polynomially bounded, there exists a polynomial-time separation oracle for the dual LP given in \ref{lp:flccc+dual}.
\end{lemma}

\begin{proof}
For the constraints in second and third line in the above dual, we can simply check whether they are satisfied by iterating over them. Therefore, we focus on the constraints for the stars since there are an exponential number of them: for all facilities, $i$, all stars $(i,S) \in \mathcal{S}_i$ satisfy $\sum_{j \in S} \alpha_j > \bar{c}$. We consider two cases:

\noindent \textbf{Case 1: \( i \in F_{pre} \).} Here only singleton stars of the form \( (i, \{j\}) \) are allowed, subject to the condition that \( d(i,j) \leq \min_{j' \in S_h} d(i,j') \) for the corresponding guessed tuple \( (i, S_h, g_h) \in S_{pre} \). For such facilities, the oracle simply checks whether any feasible singleton violates the constraint \( \alpha_j \geq \bar{c}(i, \{j\}) = d(i,j) \). Since only $O(n)$ such singletons need to be checked per facility, this case can be handled in linear time.

\noindent \textbf{Case 2: \( i \notin F_{pre} \).} Here, the oracle searches for a subset \( S \subseteq C \setminus C_{pre} \) such that $\sum_{j \in S} \bar{d}(i,j) \leq G - \bar{f}_i$ and where $
\sum_{j \in S} \alpha_j > \bar{c}'(i,S)$. This is precisely a knapsack problem: each client \( j \) has a profit \( \alpha_j - \bar{d}(i,j) \), a weight \( \bar{d}(i,j) \), and the knapsack budget is \( G - f_i \). We check if there is a set of clients that has profit at least $\Bar{f}_i$. Since the costs \( \bar{d}(i,j) \) and \( \bar{f}_i \) are integral and polynomially bounded (as ensured in Stage 0), this problem can be solved in polynomial time via dynamic programming.

Using the ellipsoid method with this separation oracle, we solve the dual LP optimally where only polynomially many constraints are tight. Each such tight constraint corresponds to a possibly non-zero primal variable, i.e., a residual star $(i, S)$ with nonzero weight in the dual solution. These define a polynomial-sized support, which we now use to construct a restricted primal LP. Since the number of stars in this support is polynomial, we can write down and solve the primal LP over this restricted set system explicitly, obtaining a fractional solution $(\vecx^*, \vecz^*)$. 
\end{proof}

Note that the LP may not be feasible if our guess was incorrect. In this case we can discard the guess. For a correct guess, we denote the cost of this LP solution by $\OPT_{\LP}$. Note that $\OPT_{\LP} \le \OPT - \OPT_{guess}$.

% From this point forward, we treat the polynomial-sized support of stars identified by the dual as our effective CCF set system, and apply the remaining stages of the $\CCF$ framework over this restricted instance.

\subparagraph*{Stage 2': Handling Heavy vs. Shallow Clients.}
Given the fractional solution $(\vecx^*, \vecz^*)$ to the residual LP, we now have a polynomial-sized support of stars. We begin by partitioning the clients into \emph{heavy} clients $C_{he}$ and \emph{shallow} clients $C_{sh}$ based on their fractional coverage: let $C_{he} := \{ j \in C \setminus C_{pre} \mid z_j > \tau \}$ and $C_{sh} := \{ j \in C \setminus C_{pre} \mid z_j \leq \tau \}$, where we set $\tau := (1-\sfrac{1}{e})(1-\eps)$ as in the $\CCF$ framework.

%We first fully cover the heavy clients $C_{he}$ by solving an instance of the canonical facility location problem restricted to just these clients. We do this by using a $\beta_{FL}$-approximation algorithm for facility location that operates over the standard LP formulation. Using \cite{li20131}, we can obtain a $\beta_{FL} = 1.488$ approximation. Recall, however, that the discretized distances $\bar{d}$ used in our residual LP may no longer satisfy the triangle inequality. Therefore, in this step we revert to the original metric $d$ to ensure that the assumptions of the approximation algorithm hold. The stars retained after Stage 2 form our feasible facility-client pairs, and their facility costs $f_i$ can also be reverted to the original, non-discretized distances \todo{however we update these original costs to reflect that guessed facilities are already payed for in Stage 1, that is, for $i \in F_{pre}$, let $f'_i = 0$, otherwise $f'_i = f_i$}. Let $\mathcal{S}_{he}$ denote the stars that correspond to the approximate solution obtained using the standard facility location LP. The following lemma (which is analogous to \Cref{lem:ccf-heavy} in the $\CCF$ analysis) will hold.
To cover the heavy clients, we prove the following lemma.
\iffalse
\begin{lemma}\label{lem:flccc-heavy}
The total cost of $\mathcal{S}_{he}$ is at most 
\(
\beta_{\text{FL}} \cdot \left(\frac{e}{e-1}\right) \cdot \left(\frac{1}{1 - \epsilon} \right) \cdot \OPT_{{\LP}},
\)
where $\OPT_{{\LP}}$ denotes the cost of the components of the optimal solution that we have not guessed in the optimal solution, and $\beta_{FL}$ is the approximation factor of the facility location algorithm applied to cover the heavy clients $C_{he}$.
\end{lemma}
\fi

\begin{lemma}\label{lem:flccc-heavy}
Given $(\vecx^*, \vecz^*)$, a feasible solution to the residual LP with cost $\OPT_{{\LP}}$, there is an efficient algorithm 
to cover the heavy clients $C_{he}$ with cost at most 
\(
\beta_{\text{FL}} \cdot \left(\frac{e}{e-1}\right) \cdot \left(\frac{1}{1 - \epsilon} \right) \cdot \OPT_{{\LP}},
\)
where $\beta_{FL}$ is the approximation factor of the underlying LP-based approximation algorithm for UCFL.
\end{lemma}

\begin{proof}
    The standard LP for UCFL is as follows where $y_i$ is an indicator for facility $i$ being opened and $x_{ij}$ is the indicator for client $j$ connecting to facility $i$:
    \begin{equation}\label{lp:facility location}\tag{$\LP$-$\textsf{FL}$}
 \begin{array}{r@{}c@{}l}
    & \min \sum_{i \in F}y_if_i + \sum_{j \in C, i \in F} x_{ij}d_{ij}  \\ 
\text{s.t.}     & \sum_{i \in F} x_{ij} \geq 1 &  \text{for all } j \in C\\ 
  &  y_i  \geq   x_{ij} & \text{for all } i \in F \text{ and } j \in C \\ 
   & x_{ij} \geq 0, y_i \geq 0  & \text{for all } i \in F \text{ and } j \in C
\end{array}
\end{equation}
    Using the fractional solution obtained by solving \ref{lp:flccc+dual}, we can obtain in a canonical manner a fractional solution for the standard LP formulation \ref{lp:facility location}. Specifically, for each client $j$ and facility $i$, $x_{ij} = \sum_{S: j \in S}x(i,S)$ and for facility $i$, $y_i = \sum_{S} x(i,S)$. Note that this does not change the cost of the LP solution. Further, scaling this converted solution by $\scaleccf$ gives a feasible solution to the standard LP relaxation with a cost of $\scaleccf \OPT_{\LP}$ to connect all the clients in $C_{he}$. We can then use the best known approximation that is based on the \ref{lp:facility location}; we have $\beta_{FL}= 1.488$ by \cite{li20131}. We note here that the rounding algorithm requires the distances form a metric. Since our $\LP$ solution is obtained using the scaled and discretized distances, the assumption is not immediately satified. However, this can be easily fixed. For using the rounding algorithm, we revert the distances to the original (but still scaled) distances. The distances that increase are the ones we rounded down to $0$. These were all at most $\sfrac{\OPT}{B} = \sfrac{\OPT}{n^3}$. Therefore, in the final rounding solution they will only cost us at most $\sfrac{\OPT}{n}$ increase. This extra cost is a $(1+o(1))\OPT$ which can be absorbed in the eventual approximation.
\end{proof}

\noindent Now, to proceed with Stages $3$ and $4$ for the shallow clients, we perform the following steps.
\begin{itemize}
    \item \textbf{Restrict the instance to shallow clients.} In our fractional solution for the primal in \ref{lp:flccc+dual}, we update the stars to only have shallow clients and remove all the heavy clients. The variables $x(i,S)$ remain unchanged. The $z$ variables are restricted to only shallow clients. We also update the covering requirements to reflect that heavy clients have already been selected: $b_k' := b_k' - |C_k \cap C_{he}|$ for each color class $k \in [r]$.
    \item \textbf{Scale the solution}. We scale the solution by $\scaleccf$. Note that the new scaled solution is a feasible solution for the primal in \ref{lp:flccc+dual} with a cost $\scaleccf \cdot \OPT_{\LP}$ and each color's residual requirement is over-covered by a factor of $\frac{e}{(e-1)(1-\eps)}$.
\end{itemize}

\subparagraph*{Stage 3: Rounding the fractional solution.}
We work with the residual instance restricted to shallow clients and scaled as mentioned in the previous stage. We now want to apply \Cref{lem:additive-bound} to round this fractional solution. The algorithm in the lemma potentially picks $L$ stars to ensure that randomized rounding satisfies the constraints with high probability. Algorithmically, we still perform the same step. For each color class $k \in [r]$, we select (up to) the top $\frac{1}{\ell} \ln(\sfrac{1}{\epsilon})$ stars \emph{from the support} of the fractional solution $\vecx^*$. The proof of \Cref{lem:additive-bound} easily extends to only picking from the support.
%Note that \Cref{lem:additive-bound} was selecting sets from the entire ground set, however since the randomized rounding only chooses stars with non-zero $x(i,S)$ value, we can %also choose the largest marginal contributors from this restricted set. Since these are polynomially bounded, this can be done in polynomial time. 
Finally, we randomly round the remaining stars with probabilities given by $\vecx^*$. Suppose the set of stars chosen in this stage is $\Sets(C_{sh})$. We can prove the following key lemma.

\begin{lemma}\label{lem:fl-shallow}
    After selecting the stars $\Sets(C_{sh})$, each constraint is satisfied with probability at least $1 - \sfrac{\epsilon}{r}$ and the expected cost of $\Sets(C_{sh})$ is bounded by $\scaleccf \OPT_{\LP} + \epsilon OPT + \OPT_{guess}$.
\end{lemma}
\begin{proof}[Proof Sketch.]
    The proof for bounding the cost of sets from randomized rounding and the probability of constraints being satisfied is same as the proof of \Cref{lem:ccf-greedy-sets-bound} in $\CCF$ section. The only difference is bounding the cost of the $L$ stars that are picked by the algorithm in \Cref{lem:additive-bound}. To bound this, note that the guessed star can be of two types: (1) A singleton client that connects to a guessed facility (2) A complete star at a non-guessed facility. For the singleton stars that are at the guessed facility, since at each facility we have guessed the $L$ farthest clients and we are guessing at most $L/r$ clients per constraint, these can be accounted to the singletons guessed in Stage $1$; the total cost of the singleton clients is at most $\OPT_{guess}$. For the stars that are at a non-guessed facility, since we guessed $L/\epsilon$ largest star costs, these stars are each of cost at most $\epsilon \OPT/L$. Since we add at most $L$ such stars, the total cost of these can be bounded by $\epsilon \OPT$.
\end{proof}

\begin{remark}
    If the first stage guess shows that the total number of starts in the optimum solution is less than $L/\eps$ then we have guessed all facilities in the optimal solution. Therefore the only stars are singletons. 
\end{remark}

\subparagraph*{Stage 4: Greedily fixing remaining unsatisfied constraints.}
If any color class constraint remains unsatisfied after Stage 3, we fix it by greedily selecting stars until the constraint becomes satisfied. As in the original $\CCF$ framework, we can employ the greedy fix from the $\textsc{Max-Budgeted-Cover}$ heuristic. The \ref{lp:budgeted-covering} for this problem is to select a subset of stars that cover the maximum number of clients at a cost bounded by $\frac{e}{(e-1)(1-\eps)}\OPT_{\LP}$. Further, as per the remark after \Cref{lem:greedy-ccf-single}, this Greedy can be executed by only looking at the support of the LP solution. Therefore
the step can be performed in polynomial time. Similar to the analysis in $\CCF$, this greedy fix is applied to the $r^{th}$ covering class with a probability bounded by $\sfrac{\epsilon}{r}$. Similar to the proof of \Cref{lem:ccf-greedy-sets-bound}, the total expected cost to fix all constraints is bounded by $\epsilon \cdot \OPT + r c_{\max}$ where $c_{max}$ is the cost of highest cost star in the support. Each star in the support has cost at most $\eps \OPT/r$ due to the guessing in Stage $1$ (as discussed in proof of Lemma \ref{lem:fl-shallow}). Hence the expected fixing cost is $(1+O(\eps))\OPT$.

Putting together, we obtain the following result. The analysis is similar to that of $\CCF$. The running time is polynomial in $n^{O(r/\eps^2)}$.

\mainresultflmo*
% \begin{theorem}\label{thm:flmo}
%     There is randomized polynomial time $3.936$-approximation algorithm for the $\FLMO$ problem with fixed number of outlier classes.
% \end{theorem}

\subsection{Facility Location with Sum of Radii and Multiple Outliers}\label{sec:radii}

% Clustering problems are ubiquitious and different objectives have been considered in the literature.
% The Minimum Sum of Radii Clustering problem is one of them.
% In this problem we are given a finite metric space $(P,d)$ over $n$ points and an integer $k$. The goal is to find
% $k$-clusters to minimize the sum of the radii of the clusters. There is a generalization where one
% minimizes the sum of the $\gamma$ powers of the radii where $\gamma \ge 1$ is a fixed input parameter ($\gamma =1$ corresponds
% to sum of radii). These objectives have been well-studied and admit a constant factor approximation \cite{charikar01sum} 
% when $\gamma = 1$. A $(3+\eps)$-approximation has recently been obtained even with outliers \cite{Buchemetal23}.
% Several variants have been studied --- we refer the reader to \cite{Buchemetal23,friggstad2022improved,gibson2010metric,gibson2012clustering,charikar01sum} for pointers.

% 

In the Minimum Sum of Radii Clustering problem, we are given a finite metric space $(P,d)$ over $n$ points and an integer $k$. The goal is to find
$k$-clusters to minimize the sum of the radii of the clusters. There is a generalization where one
minimizes the sum of the $\gamma$ powers of the radii where $\gamma \ge 1$ is a fixed input parameter ($\gamma =1$ corresponds
to sum of radii). These objectives have been well-studied and admit a constant factor approximation \cite{charikar01sum} 
when $\gamma = 1$. A $(3+\eps)$-approximation has recently been obtained even with outliers \cite{Buchemetal23}.
Several variants have been studied --- we refer the reader to \cite{Buchemetal23,friggstad2022improved,gibson2010metric,gibson2012clustering,charikar01sum} for pointers.

In this section we consider the facility location version of this problem with multiple outliers, which we refer to as the Facility Location with Sum of Radii objective and Multiple Outliers ($\FLRMO$) problem; this formulation is independently interesting and
plays a role in the primal-dual algorithm for the $k$-clustering variants \cite{charikar01sum}.
The input consists of a set of facilities $F$, a set of clients $C = \cup_{k=1}^r C_k$, where $C_k$ denotes clients of color $k \in [r]$, and a metric space $(F\cup C, d)$. Each facility $i \in F$ has a non-negative opening cost $f_i$. For each color class $C_k$, we are given a coverage requirement $b_k$ where $1 \leq b_k \leq |C_k|$. Let $B(i, \rho)$ denote the subset of clients that lie within a ball of radius $\rho > 0$ that is centered at facility $i$ , i.e., $B(i, \rho) := \{ j \in C \mid d(i,j) \leq \rho \}$. 
The objective is to select a set of balls $\mathcal{B} = \{B(i,\rho_i) \mid i \in F' \subseteq F\}$, centered around a subset of facilities $F' \subseteq F$, such that (i) $\mathcal{B}$ satisfies the coverage requirement for each color class and (ii) the sum $\sum_{i \in F'} (f_i + \rho_i^\gamma)$, where $\gamma \geq 1$ and is a constant, is minimized. We refer to this problem as FL-Radii-Multi-Outliers. When there are no outliers \cite{charikar01sum} obtain a primal-dual based $3^\gamma$-approximation\footnote{As noted in \cite{chekuri2022covering}, Charikar and Panigrahy~\cite{charikar01sum}'s algorithm is for $\gamma = 1$, but with some standard ideas, one can adapt it to work for arbitrary $\gamma \geq 1$.}.
\cite{chekuri2022covering} uses this result as a black box to obtain an $O(3^\gamma + \log r)$-approximation. Although one cannot reduce it directly to CCF the high-level ideas behind CCF work. Here we consider the case when $r$ is a fixed constant.
Via the same approach as in \cite{chekuri2022covering}, we describe how to obtain the following result using our $\CCF$ framework.

\begin{theorem}\label{thm:flrmo}
    there is a randomized polynomial-time algorithm that given an instance of
    $\FLRMO$ with fixed $r$ and $\eps > 0$ yields a $(3^\gamma + 1)(\frac{e}{e-1})(1 + \eps)$-approximate solution.
\end{theorem}

We construct an instance of $\CCF$ from an instance of FL-Radii-Multi-Outliers as follows.
To construct the set system $(\Universe, \Sets)$ for our $\CCF$ instance, first set $\Universe = C$ and $\Universe_k = C_k$ for each $k \in [r]$. To define $\Sets$, we first describe for each facility $i$ a set of ``relevant'' (or feasible) radius lengths $R_i := \{ d(i,j) \mid j \in C\}$. With this, let $\Sets := \{ B(i,\rho) \mid i \in F, \rho \in R_i \}$. For each set $B(i, \rho_i) \in \Sets$, we can set the associated cost $c_{i,\rho_i} := f_i + \rho^\gamma$. The goal is to solve the following IP: %, we provide the following CCF-like IP that we must optimize: %Our goal now is to optimize the following CCF-like LP:

\begin{equation}\label{lp:mcc-multi-out}\tag{$\LP$-$\FLRMO$}
\begin{array}{cccc}
     & \min \sum_{B(i, \rho_i) \in \Sets} c_{i,\rho_i}x_{i,\rho_i} \\ 
 \text{s.t.}     & \sum_{(i, \rho_i) : j \in B(i, \rho_i)} x_{i,\rho_i} \geq z_j & \text{for all } j \in C \\ 
   &  \sum_{j \in C_k} z_j \geq   b_k & \text{for all } k \in [r]\\
   & z_j \in [0,1] & \text{for all } j \in C \\ 
    & x_{i,\rho_i} \in [0,1] & \text{for all } i \in F, \rho_i \in R_i
 \end{array}
 \end{equation}

As we noted there is a $\beta := 3^\gamma$ approximation via the LP if we wish to cover \emph{all} the clients in the universe. Note that the reason the problem is
not exactly a $\CCF$ instance is because the sets and their costs are implicitly defined by the balls and their radii. However, we mention that 
the only place we used the deletion closed property of the underlying set system is in fully covering the highly covered elements after solving the LP relaxation.
We can use the same approach here as well. The rest of the framework for $\CCF$ also applies and this yields the claimed result.

\printbibliography

\appendix 

\section{Improving Coverage for $\MSC$ Approximations} \label{sec:ext-msc}

In this section, we show how our algorithm \Cref{sec:msc} can be used to obtain a bi-criteria approximation in which coverage is improved at some expense to the cost. This result (restated below in \Cref{coll:main-result-msc}) is useful for scenarios where one wishes to achieve as close to full coverage as possible. In this setting, we are given an additional parameter $\alpha$ which helps to describe the extent of coverage required.

% least $(1-\sfrac{1}{e})$ may not be sufficient, and one wishes to achieve fuller coverage (i.e. greater than $(1-\sfrac{1}{e}$ that enhances the coverage of approximate solutions to $\MSC$ instances where there are a fixed constant number of subodular constraints.   In particular, we prove \Cref{coll:main-result-msc} (restated below) by describing an algorithm in \Cref{sec:ext-msc-alg}, we describe our algorithm, and conducting the subsequent analysis in \Cref{sec:ext-msc-analysis}. %Our analysis and the proof of , and conduct the analysis andwe Once again, we assume $r$ is a fixed constant. 

\mainresultmscgen* 

We will begin by describing the algorithm in \Cref{sec:ext-msc-alg}, and in \Cref{sec:ext-msc-analysis} we provide the subsequent analysis and prove \Cref{coll:main-result-msc}. % Note that this algorithm will not follow the framework from \Cref{sec:rf-framework}, but will instead work by repeatedly applying Algorithm \ref{alg:msc-single} (which does follow our framework) as a black-box to select elements to achieve enhanced coverage. 

\subsection{Algorithm Overview}\label{sec:ext-msc-alg}

We provide the pseudocode for our algorithm in Algorithm \ref{alg:msc-gen}, and provide more details about the various steps below. 

\begin{algorithm}[h!]
\caption{Pseudocode for Bi-criteria Approximation for $\MSC$ with improved coverage.}
\label{alg:msc-gen}
\DontPrintSemicolon
  \SetKwFunction{Define}{Define}
  \SetKwInOut{Input}{Input}\SetKwInOut{Output}{Output}
  \Input{An instance of $\MSC$ problem denoted by $\J = (N, \{ f_i, b_i \}_{i=1}^r, c)$ where $f_i$'s are normalized, parameters $\alpha \in \Nset^+$ and $\eps \in (0,1]$}
  \Output{Solution $S \subseteq N$ that satisfies \Cref{coll:main-result-msc}}
  \BlankLine
  Initialize $S := \emptyset$; rename $\J$ to $\J^{(1)}$. \label{alg:msc:init} \\ %Rename $\J$ to $\J^{(1)}$ and ensure that $f^{(1)}_i$'s are normalized. \\ \label{alg:msc:init}
  \For{$t \in [\alpha]$}{
    $S^{(t)} \gets$ result of running Algorithm \ref{alg:msc-single} on $\J^{(t)}$ using $\eps' := \min(1-\sfrac{2}{e}, \eps)$ \label{alg:msc:call-single} \\ % \singleMSC(\J^{(t)}, \eps^{(t)})$, where $\eps^{(t)} \in (0,1]$ \\ 
    $S \gets S \cup S^{(t)}$ \label{alg:msc:upd-S} \\ 
    $\J^{(t+1)} \gets$ residual instance of $\J$ with respect to $S$ (per \Cref{def:residual-msc})\label{alg:msc:upd-inst} \\% ; we can truncate each $f_i^{(t+1)}$ so that $f_i^{(t+1)}(N^{(t+1)}) = b^{(t+1)}$. \label{alg:msc:upd-inst}
    
     %per \Cref{def:residual-msc-gen}) of $\J^{(t)}$ after elements of $S^{(t)}$ are marked as selected.  \\ \label{alg:msc:upd-inst} % Ensure that the functions $f^{(t+1)}_i$'s are normalized.
   
  }
  \Return{S} \label{alg:msc:return}
\end{algorithm}

% We are given as input an arbitrary instance $\J$ of $\MSC$ as well as a parameters $\alpha \in \Nset^+$ and $\eps \in (0,1]$, which are used to describe how much coverage we require for our final solution. This algorithm accumulates a set $S$ over $\alpha$ calls to Algorithm \ref{alg:msc-single} ($\singleMSC$) that will ultimately satisfy the constraints stated in \Cref{coll:main-result-msc}. 

\paragraph{Line \ref{alg:msc:init}: Initializing the accumulator $S$.} We initialize a set $S$ to serve as the accumulator for the final solution. We also rename $\J$ to  $\J^{(1)}$ for clarity in the algorithm description and analysis. Also, like in \Cref{sec:msc} and without loss of generality, we may enforce that all functions in the instance are normalized, i.e., $f^{(1)}_i(N^{(1)}) = b^{(1)}_i$.  % We will use $\J^{(1)}$ to refer to the input instance (for clarity in the analysis) and 

\paragraph{Lines \ref{alg:msc:call-single}-\ref{alg:msc:upd-S}: Choosing elements for $S$ in $\alpha$ rounds.} We do the following procedure $\alpha$ times. In each round $t$, we use Algorithm \ref{alg:msc-single} on the instance $\J^{(t)}$ with parameter $\eps' := \min(1-\sfrac{2}{e}, \eps)$. % and for some $\eps^{(t)} \in (0,1]$. Note that we will differ our discussion about how $\eps^{(t)}$'s are to be set to the analysis, as it will be relevant to analyzing the runtime. 
The call to Algorithm \ref{alg:msc-single} will return a set of elements $S^{(t)}$ that will obtain $(1-\sfrac{1}{e}-\eps')$ coverage for functions in the instance $\J^{(t)}$, at cost $(1 + \eps')\OPT^{(t)}$, where $\OPT^{(t)}$ denotes the cost of the optimal solution to $\J^{(t)}$. After adding $S^{(t)}$ to $S$, we construct a residual instance $\J^{(t+1)}$ from $\J$ using the definition from \Cref{def:residual-msc}. We additionally normalize (or truncate) the restricted functions $f^{(t+1)}_i$ so $f^{(t+1)}_i(N^{(t+1)}) = b^{(t+1)}_i$. For $t < \alpha$, $\J^{(t+1)}$ is used to select elements to add to $S$ in the next round. 

% \begin{definition}[Residual Instance $\J^{(t+1)}$ of $\MSC$]\label{def:residual-msc-gen}

% Suppose we have an original instance $\J = \J^{(1)} = (N, \{f_i, b_i\}_{i \in [r]}, c)$. A residual instance $\J^{(t+1)} = (N^{(t+1)}, \{ f^{(t+1)}_i, b^{(t+1)}_i\}_{i=1}^r, c^{(t+1)})$, constructed from a restricted instance $\J^{(t)}$ with respect to a set $X \subseteq N^{(t)}$ and set $S \subseteq N$ of elements taken, is defined as follows
% \begin{itemize}
%     \item  $N^{(t+1)} = N^{(t)} - X$,
%     \item $f^{(t+1)}_i(A) = f^{(t)}_{i\mid X}(A)$, $\forall A \subseteq N^{(t+1))}$ and $b^{(t+1))}_i = b_i - f_i(S \cup X)$ $\forall i \in [r]$, and 
%     \item $c^{(t+1)} = c^{(t)}_{\mid N^{(t)}}$
% \end{itemize}
% % , and a set 
% %       Given an instance $\J = (N, \{f_i, b_i\}_{i \in [r]}, c)$ and a set $X \subseteq N$ of elements (which we intend to include as part of a solution for $\J$), we define a residual instance $\J^{(\cdot)} = (N^{(\cdot)}, \{ f^{(\cdot)}_i, b^{(\cdot)}_i\}_{i=1}^r, c^{(\cdot)})$ as follows: 
% %  $N^{(\cdot)} = N - X$; $f^{(\cdot))}_i(A) = f_{i\mid X}(A)$, $\forall A \subseteq N^{(\cdot))}$ and $b^{(\cdot))}_i = b_i - f_i(X)$ $ \forall i \in [r]$; and $c^{(\cdot)} = c_{\mid X}$. 
%  Recall that we may assume that $b^{(t+1)}_i \geq 0$ (since the value of $f_i$ will never exceed $b_i$ per our initial assumptions/truncation). Also, $f^{(t+1)}_i$ may be similarly truncated to ensure that $f^{(t+1)}_i(N^{(t+1)}) = b^{(t+1)}_i$.
% \end{definition} 

% that For this algorithm, we use a slightly adjusted definition of restricted instance, which we state below. 

The algorithm will output $S = \cup_{t=1}^\alpha S^{(t)}$). In the next section, we show how this obtains the desired coverage and cost bounds.

\subsection{Analysis}\label{sec:ext-msc-analysis}

Recall that in Algorithm \ref{alg:msc-single}, we only needed to construct a single cost-truncated residual instance of the input instances $\J$. Now, in Algorithm \ref{alg:msc-gen}, we are constructing $\alpha$ different residual instances of $\J$. As such, we first prove the following simple claim. 

\begin{claim}\label{clm:res-msc-gen}
    Let $\J$ be an instance of $\MSC$ whose optimal solution has cost $\OPT$. Suppose $\J'$ is the residual instance of $\J$ with respect to some $X \subseteq N$, and let $\OPT'$ denote the cost of its optimal solution. Then, $\OPT' \leq \OPT$. 
\end{claim} 

\begin{proof}
    Let $S^*$ denote an optimal solution to $\J$, and let $Y = S^* \setminus X$. We claim that $Y$ is a feasible solution for $\J'$, since it contains all the elements that satisfy constraints in $\J$ and thus will satisfy any remaining unsatisfied constraints for $\J'$. Since $Y$ is a feasible solution for $\J'$, then $\OPT' \leq c(Y) \leq c(S^*) = \OPT$, which proves the claim. % Thus, we know that $\OPT'$ will be at most $c'(Y )$ (which is equal to $c(Y)$). Hence, since $\OPT' \leq c(Y)$ and $Y \subseteq S^*$, then $\OPT' \leq c(Y) \leq c(S^*) = \OPT$.  
\end{proof}

We can now prove \Cref{coll:main-result-msc}.

\begin{proof}[Proof of \Cref{coll:main-result-msc}]
    We proceed via induction, and by first assuming that $\eps \leq 1 - \sfrac{2}{e}$ (later we show why this assumption is not actually necessary). More precisely, for all $\alpha \in \Nset^+$, we will prove the following claim which we denote as $P(\alpha)$: for $\alpha$ and a fixed $\eps \in (0,1-\sfrac{2}{e}]$, given an instance of $\MSC$,  Algorithm \ref{alg:msc-gen} will return a solution $S$ s.t. $\forall i \in [r]$, $f_i(S) \geq (1-e^{-\alpha} - \eps)b_i$ and $\E[c(S)] = \alpha(1 + \eps)\OPT$, where $\OPT$ denotes the cost of the optimal solution to $\J$. 
    
    Since we are assuming that $\eps \leq 1 - \sfrac{2}{e}$, we know that $\eps'$ in Line \ref{alg:msc:call-single} will be set to $\eps$. Our base case, i.e. when $\alpha = 1$, is equivalent to \Cref{thm:main-result-msc}, since Algorithm \ref{alg:msc-gen} will simply run one iteration of $\singleMSC$ before returning $S = S^{(1)}$. 
    
    Assume that for all $t$ s.t $1 \leq t < \alpha$, $P(t)$ holds. %More precisely, for some fixed $\eps_{prev} \in (0,1]$ Algorithm \ref{alg:msc-gen} will return a solution $S$ s.t. $\forall i \in [r]$, $f_i(S) \geq (1-e^{-t} - \eps_{prev})b_i$ and $\E[c(S)] = t(1 + \eps_{prev})\OPT$. Here, we use $\eps_{prev}$ to emphasize that the result of the previous round need not be for $\eps$ but can be for an intermediary parameter $\eps_{prev}$.   %s will actually be some function of the parameters $\eps^{(t)}$ chosen for each call to Algorithm \ref{alg:msc-single} in previous rounds $1$ to $\alpha-1$ of the for loop. 
    We will now show that $P(\alpha)$ holds. Consider the for-loop in Algorithm \ref{alg:msc-gen}, and in particular consider the state of the algorithm after the for loop has run for $\alpha-1$ rounds. At this point, $S = \cup_{t=1}^{\alpha-1} S^{(t)}$, and $\J^{(\alpha)}$ has been initialized as the residual instance of $\J$ with respect to this $S$. Per our induction hypothesis, we know that $S$ satisfies the claim $P(\alpha-1)$. Now we are left to show that $S \cup S^{(\alpha)}$ where $S^{(\alpha)}$ will be constructed in the final round of the for loop, satisfies the claim $P(\alpha)$. 
    
    First, we show that $S \cup S^{(\alpha)}$ achieves sufficient coverage:
    \begin{align*}
        f_i(S \cup S^{(\alpha)}) &= f_i(S) + f_i(S^{(\alpha)} \mid S) \\
        &=  f_i(S) +f_i^{(\alpha)}(S^{(\alpha)}) \\ 
        &\geq f_i(S) + (1-\sfrac{1}{e}-\eps)b^{(\alpha)}_i &&\text{(by \Cref{thm:main-result-msc})}\\
        &= f_i(S) + (1-\sfrac{1}{e}-\eps)(b_i-f_i(S)) \\
        &= (1-\sfrac{1}{e}-\eps)b_i + (\sfrac{1}{e} + \eps)f_i(S) \\
        &\geq (1-\sfrac{1}{e}-\eps)b_i + (\sfrac{1}{e} + \eps)(1-e^{1-\alpha} - \eps)b_i &&\text{(by induction hypothesis)}\\
        &= (1 - e^{-\alpha} - (\sfrac{\eps}{e} + e^{1-\alpha}\eps + \eps^2))b_i \\
        &\geq (1 - e^{\alpha} - \eps)b_i 
    \end{align*}

   % The first inequality is due to the guarantee given in \Cref{thm:main-result-msc} for a single iteration of $\singleMSC$. The second inequality comes from utilizing the induction hypothesis. 
    The last inequality follows from the fact that $\eps \geq \sfrac{\eps}{e} + e^{1-\alpha}\eps + \eps^2$ holds when $\eps \leq 1 - \sfrac{1}{e} - e^{1-\alpha}$, which is at most $1 - \sfrac{2}{e}$ when $\alpha \geq 2$. %  (and where the analysis is equivalent to the above). Since $\eps' \leq \eps$, 

    %will hold when $\eps_{prev}$ and $\eps^{(\alpha)}$ are set appropriately, i.e., when they are set s.t. \[\eps \geq \sfrac{\eps_{prev}}{e} + e^{1-\alpha}\eps^{(\alpha)} + \eps_{prev}\eps^{(\alpha)}\] and $\eps_{prev},\eps^{(\alpha)} < \eps$ (such a setting is possible when $\eps \leq 1 - \sfrac{1}{\eps}$ and $\alpha \geq 2$). See \Cref{rem:eps} for a more detailed explanation of how to set $\eps^{(t)}$ and $\eps_{prev}$ for intermediary rounds to ensure that by the final round, we have coverage and cost wrt to $\eps$ from the problem input. % More detailed backtracking will illuminate how to set each $\eps^{(t)}$ in each round to achieve the desired coverage. % acheive coverage with respect to the original $\eps$. coverageorder to ensure that  $  It is important to note that $\eps'$ will also be affected by depend on previous choices of $\eps^{(t)}$'s, and thus one must be careful in choosing $\eps^{(t)}$ at each round.

    % $\eps' \in O(\eps)$. The first equality holds since $f_i^{(\alpha)}$ is the marginal value of $f_i$ given that elements of $S$ are chosen (note that this is because $\J^{(\alpha)}$ can not only be thought of as the residual instance of $\J^{(\alpha-1)}$ once $S^{(\alpha-1)}$ is selected, but also as the residual instance of $\J$ once $S$ is selected). From there, we utilize \Cref{def:residual}, \Cref{coll:msc}, and our induction hypothesis to achieve the final result. 

    Next, we show that desired cost bounds are also maintained. Let $\OPT$ be the optimal cost of $\J$, and $\OPT^{(t)}$ denote the optimal cost of $\J^{(t)}$. Per \Cref{clm:res-msc-gen}, we know that $\OPT^{(t-1)} \geq \OPT^{(t)}$. Now, per our induction hypothesis we know that $\E[c(S)] \leq (\alpha-1)(1+\eps)\OPT$. Similarly, by applying \Cref{thm:main-result-msc} to the final residual instance $\J^{(\alpha)}$, we know that $\E[c^{(\alpha)}(S^{(\alpha)})] \leq (1 + \eps)\OPT^{(\alpha)}$ (where $\OPT^{(\alpha)}$ denotes the cost of the optimal solution for residual instance $\J^{(\alpha)}$). Thus, we have
    \begin{align*}
        \E[c(S \cup S^{(\alpha)})] &= \E[c(S)] + \E[c^{(\alpha)}(S^{(\alpha)})] \\
        &\leq (\alpha-1)(1+\eps)\OPT + (1+\eps)\OPT^{(\alpha)} \\
        &\leq \alpha(1 + \eps)\OPT 
    \end{align*}

    This shows that the desired coverage and cost bounds are achieved when $\eps \leq 1 - \sfrac{2}{e}$. We now discuss the case where $\eps > 1 - \sfrac{2}{e}$. Notice that in Line \ref{alg:msc:call-single}, we run Algorithm \ref{alg:msc-single} using $\eps' = \min(1-\sfrac{2}{e}, \eps)$. In other words, if $\eps$ is too large, we run the algorithm with a smaller $\eps'$, in which case the above analysis goes through. In doing so, our algorithm will return a set $S$ where $\forall i \in [r]$, $f_i(S) \geq (1 - e^{-\alpha} - \eps')b_i \leq (1 - e^{-\alpha} - \eps)b_i$ and $\E[c(S)] \leq \alpha(1 + \eps')\OPT \leq \alpha(1 + \eps)\OPT$, and thus still satisfies the claim.

\end{proof}

\end{document}